\documentclass[a4paper, 11pt]{scrartcl}
	\usepackage[a4paper,left=3cm,right=3cm,top=3cm,bottom=4cm]{geometry}
	\usepackage{url}	%Internetquellen in Literaturverzeichnis 
	\usepackage{caption}
	\usepackage{enumitem} % gute Nummerierung
	\usepackage[utf8]{inputenc} % falls Umlaute in den Quellen verwenden sollen
	\usepackage{amsmath}   % enthaelt nuetzliche Makros fuer Mathematik
	\usepackage{amsthm}    % fuer Saetze, Definitionen, Beweise, etc.
	\usepackage{amsfonts}  % spezielle AMS-Mathematik-Fonts
	\usepackage{graphicx}  % erlaubt einfaches Einbinden von Graphiken
	\usepackage{float}	 % erlaubt es (unsch\"on) floats an eine bestimmte Stelle zu zwingen
	\usepackage{mathrsfs} % damit bekomme ich geschwungene Buchstaben
	\usepackage{amssymb} %für semidirekte Produkte
	\usepackage{mathtools} % da ist ein definiert als zeichen bei
	\usepackage{microtype} % soll overfulls im Literaturverzeichnis beheben   
	\usepackage[square, numbers]{natbib}

	\usepackage[ruled]{algorithm2e}
	\usepackage{algpseudocode} 
	\usepackage{blindtext}
	\usepackage{scrextend}
	\usepackage[colorinlistoftodos]{todonotes}  
	\usepackage{thmtools}
	\usepackage{thm-restate}
	\usepackage[hidelinks]{hyperref}
	\usepackage{cleveref}
	
	\usepackage[utf8]{inputenc}
	\usepackage[acronym,numberedsection,section=subsection]{glossaries}
\makenoidxglossaries
 
\newglossaryentry{anterograde}
{name=anterograde,
        description={Direction of transport from the soma to the tip of a neurite}}
 
\newglossaryentry{endocytosis}
{ name=endocytosis,
        description={Internalization of an area of cell membrane as a vesicle}}
 
\newglossaryentry{exocytosis}
{name=exocytosis,
        description={Insertion of a vesicle into the plasma membrane}}
 
\newglossaryentry{GFP}
{name=GFP,
        description={Green fluorescent protein}}
        
\newglossaryentry{growth cone}
{name=growth cone,
        description={Dynamic cellular structure at the tip of neurites that contains cytoskeletal elements and vesicles}}

\newglossaryentry{microtubul}
{name=microtubuli, plural=microtubules,
        description={Cytoskeletal elements that serves as tracks for intracellular transport}}
        
\newglossaryentry{polarization}
{name=polarization,
        description={Establishment of an asymmetric organization of cells}}
        
\newglossaryentry{progenitor cell}
{ name=progenitor cell,
        description={Stem cell that generates neurons by cell division}}
 
\newglossaryentry{retrograde}
{name=retrograde,
        description={Direction of transport from the tip of a neurite to the soma}}
 
\newglossaryentry{soma}
{name=soma,
        description={Cell body of neuron}}
        
\newglossaryentry{Vamp2}
{name=Vamp2,
        description={Vesicular membrane protein}}

\newglossaryentry{Vamp2-GFP}
{name=Vamp2-GFP,
        description={Fusion protein of Vamp2 to GFP}}
        
\newglossaryentry{vesicle}
{name=vesicle,
        description={Organelle separating its contents from the cytoplasm by a membrane (lipid bilayer)}}

	\addtokomafont{labelinglabel}{\sffamily}

  \newcommand{\R}{\ensuremath{\mathbb{R}}}   % reelle Zahlen
  \renewcommand{\epsilon}{\varepsilon}       % damit epsilon und phi so
  \renewcommand{\phi}{\varphi}		         % aussehen wie gewohnt	
  \newtheorem{satz}{Satz}[section]
  \newtheorem{lemma}[satz]{Lemma}

  \newtheorem{remark}[satz]{Remark}
  \theoremstyle{definition}

  \theoremstyle{remark}
\begin{document}

\title{On the Role of Vesicle Transport in Neurite Growth: Modelling and Experiments}
\author{Ina Humpert\thanks{Applied Mathematics M\"unster: Institute for Analysis and Computational Mathematics,
Westf\"alische Wilhelms-Universit\"at (WWU) M\"unster, Germany (ina.humpert@uni-muenster.de).}, \ \
Danila Di Meo$^\dag$, 
\\
Andreas W. Püschel\thanks{Institute for Molecular Biology, Westfälische-Wilhelms-Universität (WWU) Münster, Germany (d\_dime01@uni-muenster.de,apuschel@uni-muenster.de)}, \ \
Jan-Frederik Pietschmann\thanks{Technische Universität Chemnitz, Fakult\"at f\"ur Mathematik, Germany (jfpietschmann@math.tu-chemnitz.de).}
}
\maketitle
\begin{abstract}
\textbf{Abstract}
\vspace{1ex}
 \\ The processes that determine the establishment of the complex morphology of neurons during development are still poorly understood. We present experiments that use live imaging to examine the role of \gls{vesicle} transport and propose a lattice-based model that shows symmetry breaking features similar to a neuron during its \gls{polarization}. In a otherwise symmetric situation our model predicts that a difference in neurite length increases the growth potential of the longer neurite indicating that 
% not only protein regulation but also 
 \gls{vesicle} transport can be regarded as a major factor in neurite growth. 

\vspace{2ex}

\end{abstract} 
\section{Introduction}

Neurons are highly polarized cells with functionally distinct axonal and dendritic compartments. These are established during their development when neurons polarize after their generation from neural \gls{progenitor cell}s and are maintained throughout the life of the neuron \cite{namba_extracellular_2015}. Unpolarized newborn neurons from the mammalian cerebral cortex initially form several undifferentiated processes of similar length (called neurites) that are highly dynamic (\cite{cooper_cell_2013}, \cite{hatanaka_excitatory_2013}). During neuronal \gls{polarization}, one of these neurites is selected to become the axon. 

The aim of this paper is to combine experimental results with modelling to better understand the role of transport in this process. Indeed, while transport of vesicles in developing and mature neurons has been studied before \cite{PhysRevE.74.031910,PhysRevLett.114.168101,10.7554/eLife.20556}, to the best of our knowledge there are so far no models that examine its impact on neuronal polarization.

For the experiments we use primary cultures of embryonic hippocampal neurons which are widely used as a model system to study the mechanisms that mediate the transition to a polarized morphology \cite{schelski_neuronal_2017}. After attaching to the culture substrate, neurons extend multiple undifferentiated neurites that all have the potential to become an axon. Before neuronal polarity is established, these neurites display randomly occurring periods of extension and retraction (\cite{cooper_cell_2013}, \cite{winans_waves_2016}). Upon \gls{polarization}, one of the neurites is specified as the axon and elongates rapidly (\cite{namba_extracellular_2015}, \cite{schelski_neuronal_2017}). This neurite has to extend beyond a minimal length to become an axon (\cite{dotti_experimentally_1987}, \cite{goslin_experimental_1989}, \cite{yamamoto_differential_2012}). 

The extension and retraction of neurites depends on cytoskeletal dynamics and the exo- and \gls{endocytosis} of \gls{vesicle}s (\cite{pfenninger_plasma_2009}, \cite{schelski_neuronal_2017}, \cite{tojima_exocytic_2015}). The growth of neurites and axons requires an increase in the surface area of the plasma membrane by the insertion of \gls{vesicle}s in a structure at the tip of the developing neurite which is called the \gls{growth cone}. Retraction, on the other hand, is accompanied by the removal of membrane through \gls{endocytosis} (\cite{pfenninger_regulation_2003}, \cite{pfenninger_plasma_2009}, \cite{tojima_exocytic_2015}). The material for membrane expansion is provided by specialized \gls{vesicle}s that are characterized by the presence of specific vesicular membrane proteins (\cite{gupton_integrin_2010}, \cite{quiroga_regulation_2018}, \cite{urbina_spatiotemporal_2018}, \cite{wang_lgl1_2011}). The bidirectional transport of \gls{vesicle}s along neurites provides material produced in the cell body and recycles endocytosed membranes \cite{lasiecka_mechanisms_2011}. Molecular motors transport organelles along \glspl{microtubul} in the \gls{anterograde} direction towards the neurite tip (kinesins) and \gls{retrograde}ly to the \gls{soma} (dynein). The nascent axon shows a higher number of organelles compared to the other neurites due to a \gls{polarization} of intracellular transport to provide sufficient material for extension (\cite{bradke_neuronal_1997}, \cite{schelski_neuronal_2017}). The net flow of \gls{vesicle}s into a neurite, thus, has to be regulated depending on changes in neurite length but it is not known how intracellular transport is adjusted to differences in the demand for \gls{vesicle}s in growing or shrinking neurites. 

Based on these findings, we aim to obtain a better understanding of the role of vesicle transport in the polarization process by means of modelling. We propose a lattice based approach for the transport of the vesicles between soma and growth cones. We model antero- and retrograde vesicles as two different types of particles that are located on a discrete lattice. To account for the limited space, we propose a maximal number of vesicles that can occupy one cell, see also Figure \ref{fig:2BScetchModelNeuronLattice}. Particles randomly jump to neighbouring cells with a rate that is proportional to a diffusion coefficient and (the relative change) of a given potential. Furthermore, only jumps into cells which are not yet fully occupied are allowed. This is closely related to so-called (asymmetric) exclusion processes, see \cite{DERRIDA199865} and the references therein. Finally, at each end of the lattice, we introduce a pool that represents either the vesicles present in the soma or, at the tip, those in the growth cones.

Lattice based models (also called cellular automata) of this type are used in many applications, ranging from transport of protons, \cite{kier2013cellular}, to the modelling of human crowds, \cite{ali2013modeling}. They also serve as a tool to understand the fundamental characteristics of systems with many particles, \cite{simpson.hughes.ea:diffusion}. Models with different species and size exclusion have also been studied, see \cite{evans1995asymmetric,Pronina_2007} and in particular \cite{PhysRevE.97.022105} which deals with the case of limited resources.

Here, we develop a model adopted to neuron polarization and present numerical simulations that analyze the relation between \gls{vesicle} transport and neurite length. We are able to show that an initial length advance of a single neurite leads to further asymmetries in the \gls{vesicle} concentration in the pools.

We also present a system of nonlinear partial differential equations that arises as a (formal) limit from the cellular model and briefly discuss its properties.

While our model still lacks major features of neurite growth, the presented results show high correspondence with real data: In live cell imaging experiments a neurite that has exceed a critical length during the \gls{polarization} process grows rapidly becoming the future axon. In our model we supply one neurite with an initial length advantage and the dynamics of the model result in a positive feedback that further increases its length advance indicating that it becomes the future axon. We also observe oscillations in the vesicle concentration in the pools that may be interpreted as cycles of extension and retraction.

This paper is organized as follows: In section 2 we explain the biological background of the paper. In section 3 we introduce a discrete model for the numeric simulations and in section 4 we present the corresponding macroscopic cross diffusion model.
In section 5 we preformed the numerical simulations and interpreted the results. 
Finally, in section 6 we give a conclusion. 

\section{Experimental Results and Consequences for the Modelling}\label{sec:biological}
This section contains the results of live cell imaging of primary neurons that were prepared from rat embryos and a brief discussion of the consequences for the mathematical modelling. The final model with all details is then presented in section \ref{sec:modelling}.
%This information will enable us to propose a discrete model describing the dynamics of \gls{vesicle}s being transported in the neurite.
%These dynamics can be portrayed by particles hopping on a lattice with certain probabilities on the microscopic point of view and we will discuss this model in full detail in section \ref{sec:modelling}.

%First we give a brief overview of the growth process itself. 

\subsection{Description of Growth Process}

Unpolarized neurons extend several undifferentiated neurites called minor neurites. Upon \gls{polarization}, one of the minor neurites is specified as the axon and growths rapidly, see Figure \ref{fig:BioFigureAB} a). 
As explained in the introduction, this increase in the length of a neurite depends on the
sum of the surfaces of all \gls{vesicle}s that fuse with the plasma membrane of the \gls{growth cone}.
Their intracellular transport from the soma to the tip of the neurite is driven by molecular motors that transport the \gls{vesicle}s along \glspl{microtubul}, see Figure \ref{fig:BioFigureAB} b).

\begin{figure}
	\begin{center}
		\begin{tabular}{cc}
			\includegraphics[width=0.35\textwidth]{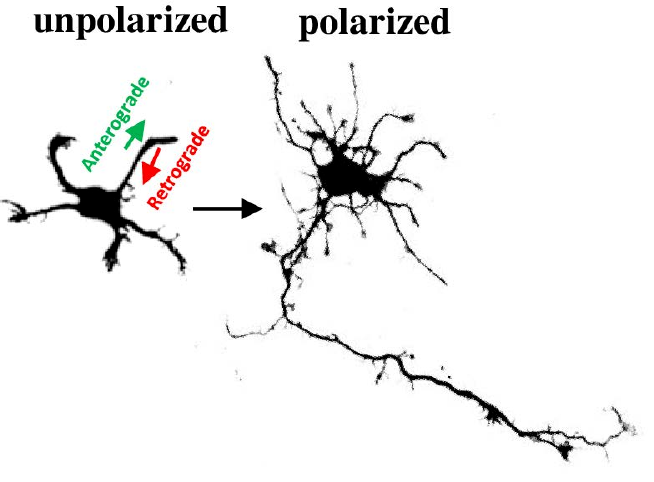} 
 			& \includegraphics[width=0.6\textwidth]{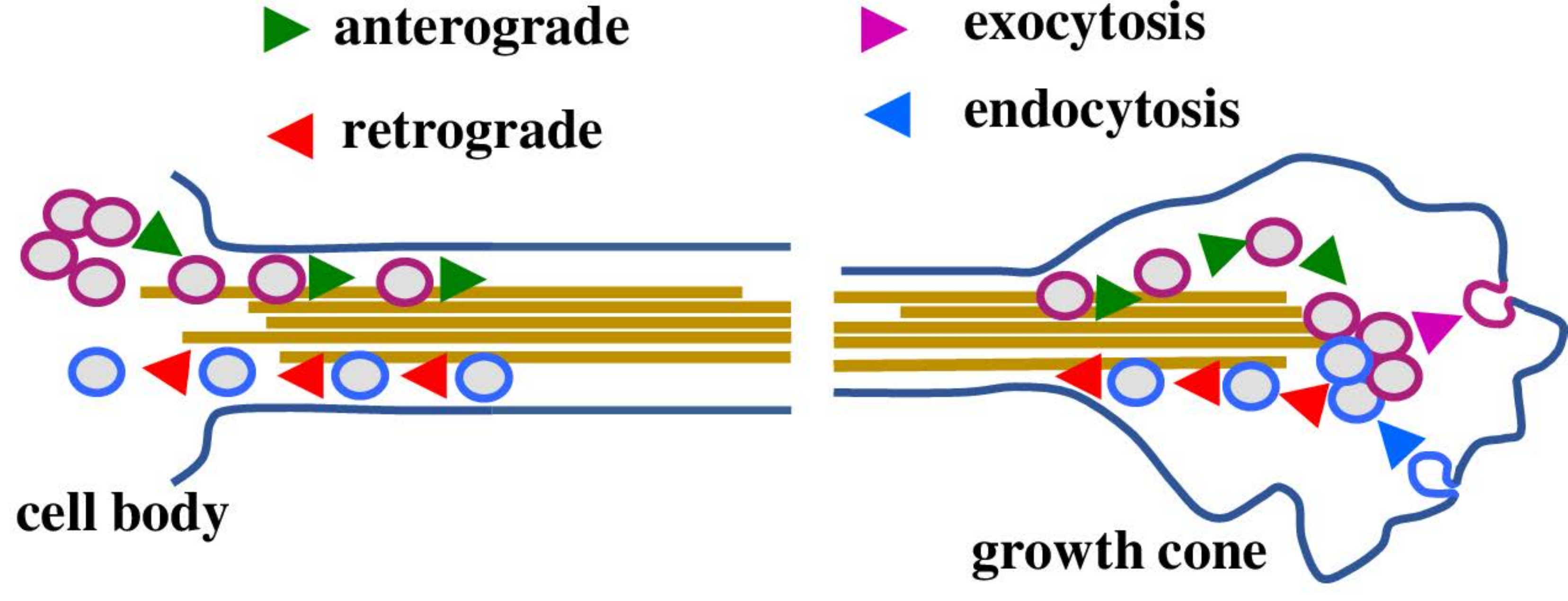} 
 			\\
			(a) &(b)
		\end{tabular}
	\end{center}
	\caption{\textbf{Vesicle transport in unpolarized neurons:}
a) Schematic representation of neuronal \gls{polarization} in primary cultures of neurons. 
b) Schematic representation of intracellular transport in neurons. Vesicles are transported in the \gls{anterograde} direction (green) from the cell body (\gls{soma}) into the neurite towards the tip of the neurite along \glspl{microtubul} (yellow). Vesicles are inserted into the plasma membrane by \gls{exocytosis} in the \gls{growth cone} at the tip of the neurite to promote extension. Vesicles generated by \gls{endocytosis} are recycled or transported \gls{retrograde}ly to the \gls{soma} (red).
}
	\label{fig:BioFigureAB}
\end{figure}

\subsection{Exterimental Methods and Results}

\subsubsection{Preparation of neurons}
Culture of primary hippocampal neurons and DNA constructs
Hippocampal neurons were isolated from rat embryos at day 18 of development (E18), plated at a density of 130,000 cells per constructs 35 mm dish ($\mu$-Dish, Ibidi) coated with poly-L-ornithine (SigmaAldrich) and cultured at $37^\circ\text{C}$ and 5\% CO$_2$ for one day in BrightCell$^\text{TM}$ NEUMO Photostable medium (Merck Millipore) with supplements. Neurons were transfected with an expression vector for Vamp2-GFP by calcium phosphate co-precipitation as described previously \cite{shah_rap1_2017}.
The \textit{pDCX-Vamp2} vector was constructed by replacing \gls{GFP} in \textit{pDcx-iGFP} (provided by U. Müller, Scripps Research Institute, La Jolla, CA, USA; \cite{franco_reelin_2011}) by a new multiple cloning site (5’- GAATTC ACTAG TTCTA GACCC GGGGG TACCA GATCT GGGCC CCTCG AGCAA TTGGC GGCCG CGGGA TCC-3’) and \gls{Vamp2-GFP} as an XbaI and BglII fragment from \textit{pEGFP-VAMP2}(addgene, \#42308), see \cite{Martinez-Arca9011}.

\subsubsection{Live cell imaging} 
Time-lapse imaging was performed in an incubation chamber one day after transfection at $37^\circ\text{C}$ and 5\% CO$_2$ using a Zeiss LSM 800 laser scanning confocal microscope (Carl Zeiss MicroImaging, Jena, Germany) and the Zeiss ZEN Software (Carl Zeiss MicroImaging). Images were taken at a frame rate of one scan per second for 2 minutes, followed by a pause of 20 minutes for 2 to 12 hours (see Figure \ref{fig:BioFigureCD}). The number, velocity and direction of \gls{vesicle} movement were quantified using the ImageJ macro toolsets \textit{KymographClear} and \textit{KymoAnalyzer} (\cite{neumann_kymoanalyzer:_2017}, \cite{mangeol_kymographclear_2016}). 

\begin{figure}[t]
	\begin{center}
		\begin{tabular}{c}
			\includegraphics[width=0.95\textwidth]{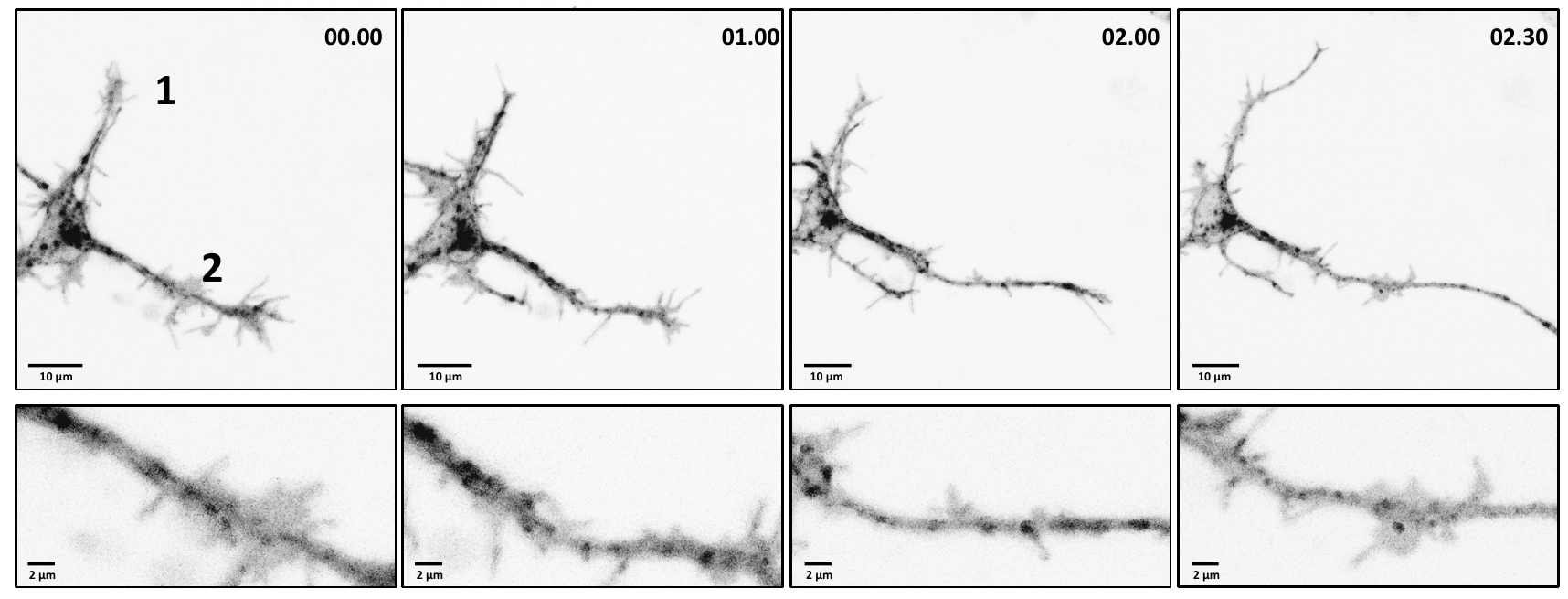} \\
			(a)
			\\
 			\includegraphics[width=0.95\textwidth]{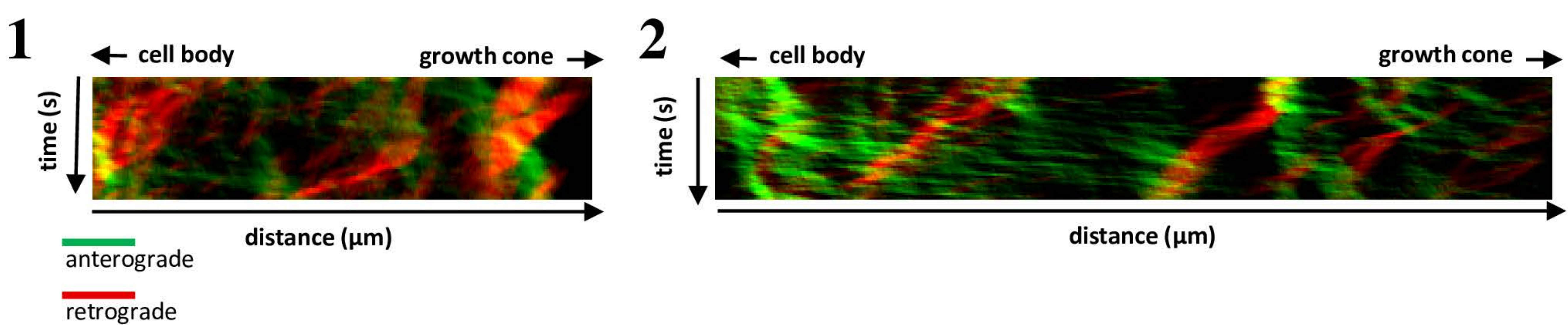} 
 			\\
 			(b)
		\end{tabular}
	\end{center}
	\caption{\textbf{Quantification of \gls{vesicle} transport in unpolarized neurons:}
a) Neurons from the hippocampus of E18 rat embryos were transfected with an expression vector for \gls{Vamp2-GFP} and the transport of Vamp2-GFP-positive \gls{vesicle}s analyzed by live cell imaging. The distribution of \gls{Vamp2-GFP}-positive \gls{vesicle}s is shown at the indicated time points (hours) in unpolarized neurons. Neurite 1 first undergoes retraction before it extends again while neurite 2 extends during the whole imaging time. A higher magnification of neurite 2 is shown in the lower panel.   
b) Representative colour-coded kymograph displaying the trajectory of moving \gls{vesicle}s in the \gls{anterograde} (green) and \gls{retrograde} (red) direction. 
}
	\label{fig:BioFigureCD}
\end{figure}
\subsubsection{Statistical analysis}
Statistical analyses were done using the GraphPad Prism 6.0 software. Statistical significance was calculated for at least three independent experiments using the Wilcoxon Sign Rank test. Significance was defined as follows: If $p > 0.05$ we regard the results as not significant and if $^\star p< 0.05$ as significant, see Figure \ref{fig:BioFigureEF} (a). 

\subsubsection{Results}
Cultures of hippocampal neurons were transfected with an expression vector for \gls{Vamp2-GFP} as a marker for \gls{vesicle}s to analyze transport in the neurites of multipolar neurons \cite{gupton_integrin_2010}. Antero- and \gls{retrograde} movement of \gls{Vamp2-GFP} positive \gls{vesicle}s was analyzed by live cell imaging 24 hours after transfection before axons are specified. To determine if the transport rates of \gls{vesicle}s change when neurites extend or retract we determined the number of \gls{vesicle}s in the neurite that are immobile or move in the antero- or \gls{retrograde} direction (Figure \ref{fig:BioFigureEF}). Vesicle dynamics is markedly higher in neurites that undergo extension or retraction compared to those that do not show changes in length. A significant difference in \gls{vesicle} transport was observed during the extension of neurites. The number of \gls{vesicle}s moving \gls{anterograde}ly (0,029 $\pm$ 0,0011 vesicles/s/$\mu \text{m}^2$) was 37 \% higher compared to those being transported \gls{retrograde}ly (0,021 $\pm$ 0,0015 vesicles/s/$\mu \text{m}^2$) in growing neurites. There was no significant difference between antero- and \gls{retrograde} transport in retracting neurites probably because not all of the \gls{vesicle}s generated by \gls{endocytosis} during retraction are positive for \gls{Vamp2}. No significant differences were found in the velocity of moving \gls{vesicle}s.

\begin{figure}[t]
	\begin{center}
		\begin{tabular}{cccc}
			(a) & \includegraphics[width=0.3\textwidth]{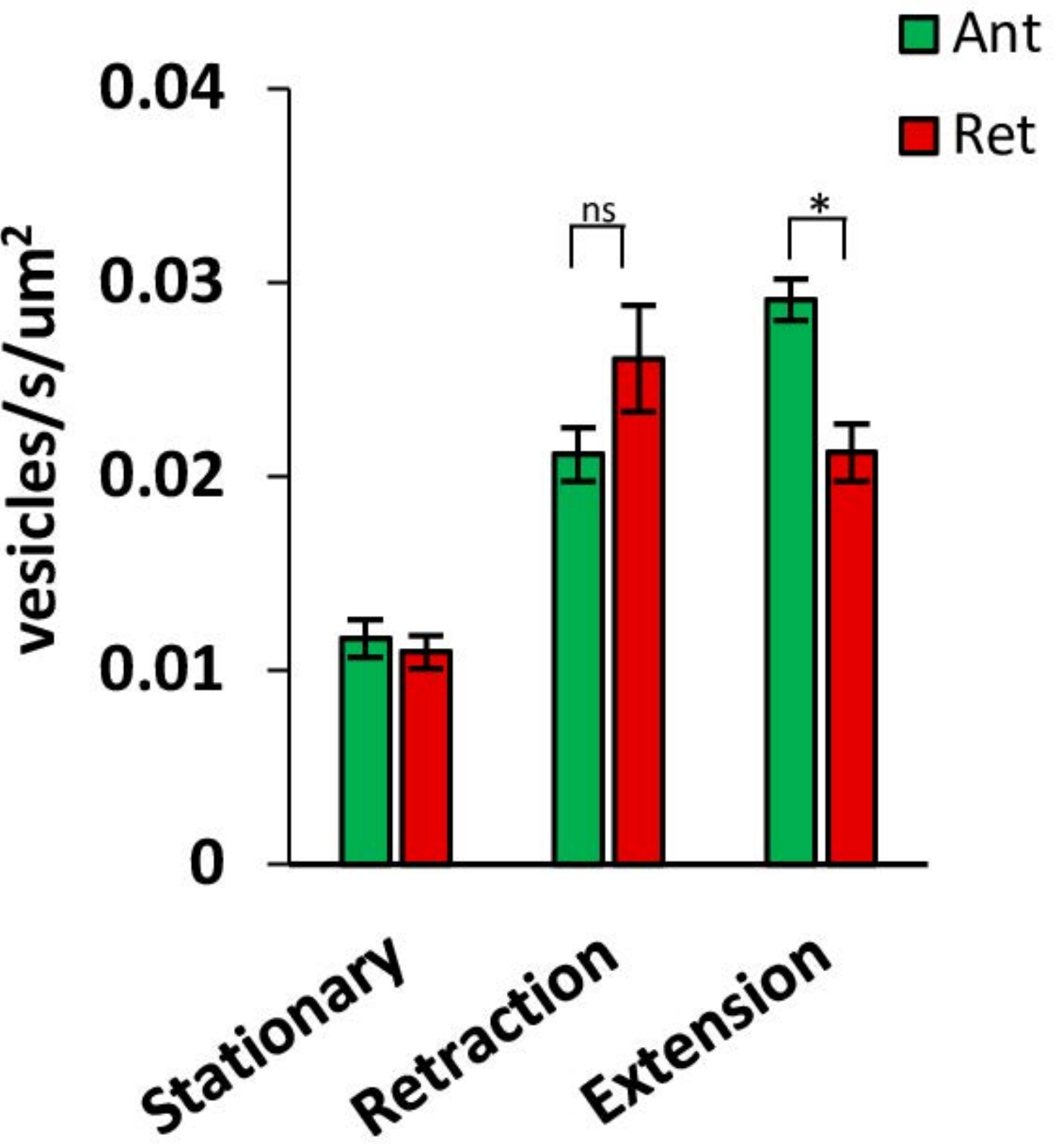} &
 			(b) & \includegraphics[width=0.3\textwidth]{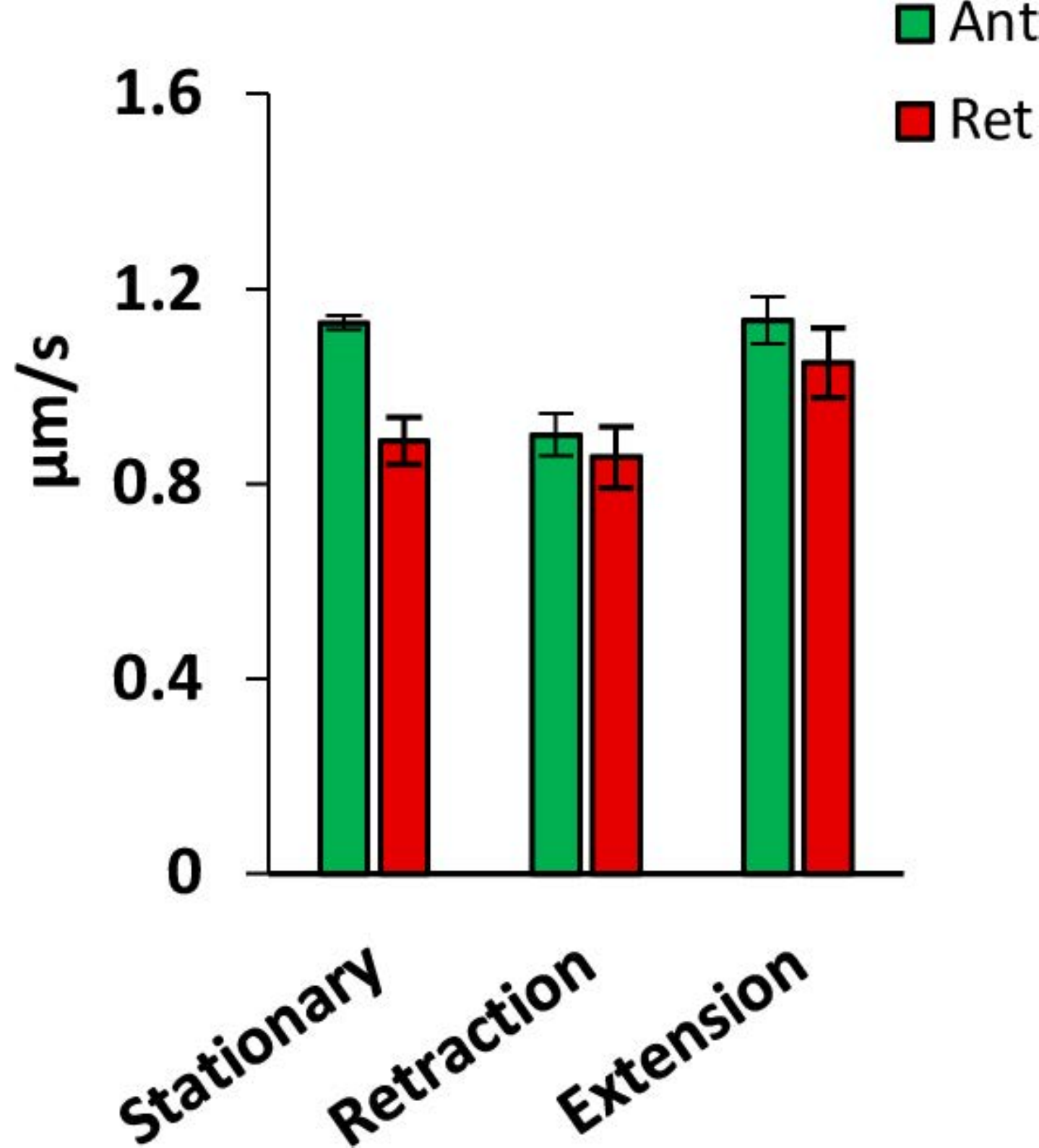} 
		\end{tabular}
	\end{center}
	\caption{\textbf{Quantification of \gls{vesicle} transport in unpolarized neurons:}
a) The number of \gls{Vamp2-GFP} positive \gls{vesicle}s moving in the \gls{anterograde} or \gls{retrograde} direction (vesicles/s/$\mu\text{m}^2$) was quantified in neurites that undergo extension or retraction and in stationary neurites that do not show a change in length. The number of moving \gls{vesicle}s is higher in neurites showing changes in length compared to stationary neurites. During extension, the number of \gls{vesicle}s moving \gls{anterograde}ly is higher (37 \%) compared to those being transported \gls{retrograde}ly (Wilcoxon Sign Rank test; $n=4$ independent experiments; values are means $\pm$ s.e.m, $^\star p< 0.05$). 
b) The speed ($\mu$m/s) of \gls{Vamp2-GFP} positive \gls{vesicle}s moving in the \gls{anterograde} and \gls{retrograde} direction was quantified. No significant differences were observed (Wilcoxon Sign Rank test; $n=4$ independent experiments; values are means $\pm$ s.e.m)} 
	\label{fig:BioFigureEF}
\end{figure}

\vspace{2ex} 

%We now summarize the relevant biological information necessary for the appropriate modeling in the following:

\subsection{Summary of the Model}\label{sec:vesicle_size}

Based on the previous experimental findings, we aim to formulate a mathematical model for the transport of vesicles based on the following assumptions: First, we consider 
neurites as one dimensional lattices connected, on one end, to the soma and to a pool representing the growth cone at the other end, see Figure \ref{fig:2BScetchModelNeuronLattice}.
Vesicles that are currently transported \gls{anterograde} and those that are moving \gls{retrograde} are modelled as two separate species moving on these lattices. At the \gls{growth cone} anterograde vesicles can fuse with the membrane while also vesicles can be separated from it. During this process, anterograde vesicles can be converted to retrograde ones and vice versa. The same can happen when vesicles enter or leave the soma.
The \gls{growth cone}s and the \gls{soma} will be modelled separately as pools that can store a given number of \gls{vesicle}s.
\begin{figure}
\centering
	\includegraphics[width=0.7\textwidth]{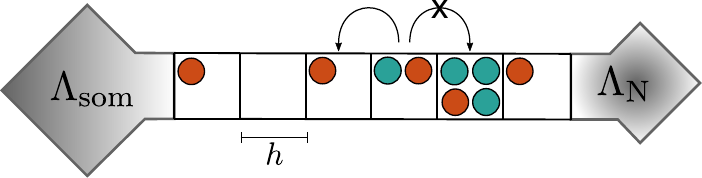}
	\caption{Sketch of the lattice based size exclusion model with the number of grid points $n=6$. For illustration purpose the maximal number of \gls{vesicle}s is 4 whereas it is higher in reality.}
	\label{fig:2BScetchModelNeuronLattice}
\end{figure}
\begin{figure}
	\centering
	\includegraphics[width=0.8\textwidth]{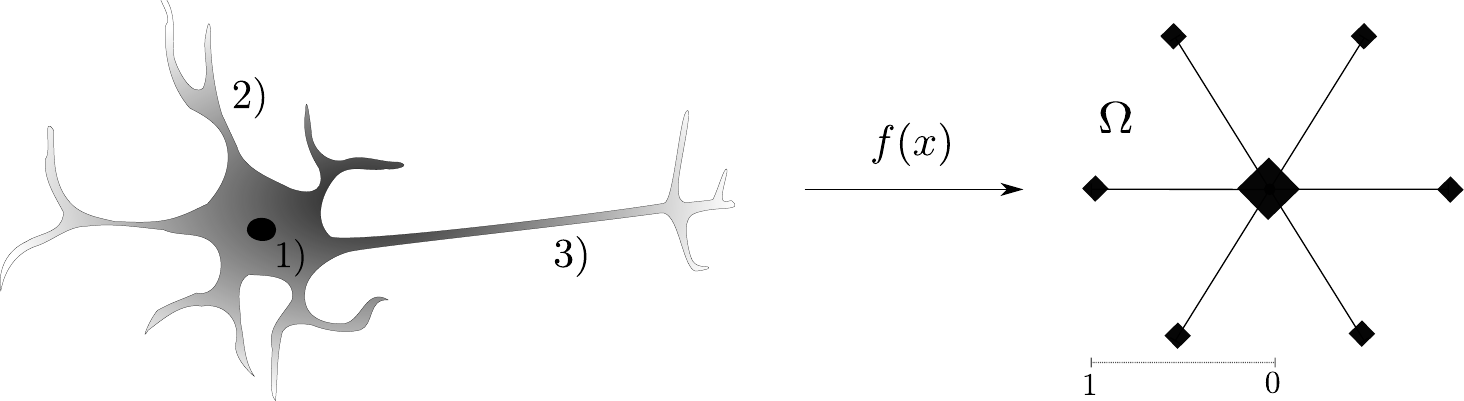}
	\caption{\textbf{Sketch of a neuron and identification with a starshaped-domain:} On the left, a sketch of a neuron can be seen, where 1) corresponds to the cell nucleus, 2) to a dendrite and 3) to the axon. On the right, a union of six unit intervals portraits the shape of the neurite that is assumed in the modelling. After the branching of the neurites is neglected, the neuron is mapped to a starshaped-domain via a function $f$.}
	\label{fig:SketchIdentificationStarObject}
\end{figure}
As \gls{vesicle}s have a positive volume, there is a maximal density within the neurites that depends on the size of the \gls{vesicle}s.
This results in a lattice model that will be described in full detail in Section \ref{sec:modelling}.

Finally, let us briefly comment of the physical dimensions involved. In practice \gls{vesicle}s with different diameters varying from $80$ to $150$ nm have been described (\cite{tsaneva-atanasova_quantifying_2009}, \cite{urbina_spatiotemporal_2018}, \cite{pfenninger_plasma_2009}) but for simplicity we assume that all \gls{vesicle}s have the same size ($130$ nm).
Thus, as the length of a neurite (which we consider as a one dimensional object) is approximately $1000$ nm, there is a natural maximal density of around $1000$ nm / $130$ nm $\approx$ 8 \gls{vesicle}s in the neurites. Another feature which depends on the diameter of the \gls{vesicle}s is the number of \gls{vesicle}s required for the extension of a neurite by a given length. %For an expansion by $10$ $\mu$m we need $\pi d$ $10$ $\mu$m = 31 $\mu$m$^2$.
%A \gls{vesicle} of $130$ nm diameter has surface area $\pi d^2$ = ($130$ nm)$^2 \pi = 53096$ $\mathrm{nm}^2$.
%Thus $584$ \gls{vesicle}s are needed for an extension of the neurite by $10$ $\mu$m.

%\begin{table}[h]
%	\begin{center}
%	\begin{tabular}{c | c | c}
%		Diameter of the Vesicle & Surface Area & Number of Vesicles Required \\
%		\hline 
%		80 nm & 20106 nm$^2$ & 1542 ves/10 $\mu$m \\
%		130 nm & 53096 nm$^2$ & 584 ves/10 $\mu$m \\
%		150 nm & 70686 nm$^2$ & 439 ves/10 $\mu$m  
%	\end{tabular}
%	\end{center}
%	\caption{Number of vesicles required for a neurite extension by 10 $\mu$m depending on the diameter of the vesicle.}
%	\label{tab:NumberofVesiclesNeuriteGrowth}
%\end{table}

\section{A discrete model}\label{sec:modelling}

We will now present a mathematical model for the growth process described in the previous section. In our approach, each neurite is modeled as a discrete lattice on which both antero- and \gls{retrograde} \gls{vesicle}s, modelled as seperate particles, move. As the diameter of a neurite is about $1000$ nm and thus very small compared to its length that can be up to $1$ m, we model neurites as one dimensional objects, i.e. a one dimensional lattice. On this lattice, the \gls{vesicle}s can jump to neighbouring cells with a probability that is determined by a given potential and a diffusion coefficient. Furthermore, we introduce a size exclusion effect by only allowing jumps to cells which are not fully occupied (see the discussion in Section \ref{sec:vesicle_size}).

These lattices are coupled to the \gls{soma} at one end and to a \gls{vesicle} pool corresponding to the \gls{growth cone} at the other end. See Figure \ref{fig:2BScetchModelNeuronLattice} for a summarized version of the model.
We will now describe the dynamics on the lattice as well as the coupling to the \gls{soma} and pools in detail, in the simple case of a single neurite connected to a \gls{soma}. \medskip\\
We first present the detailed dynamics of a single neurite:

1. \emph{Lattice dynamics}:
Each lattice consists of $i=1,\ldots,N$ cells of width $h$.
The midpoint of cell $i$ is denoted by $x_i$ and each cell can be occupied by a certain number of \gls{vesicle}s, depending on their size. Denoting by $a_i=a_i(t)$ and $r_i=r_i(t)$ the number of antero- and \gls{retrograde} \gls{vesicle}s at time $t$ in cell $i$, we have the following dynamics for the interior cells $i=2,\ldots, N-1$.
\begin{align}\label{eq:discrete_interior}
	\begin{split}
	2 C h^2 \partial_t a_i &= - a_i (1-\rho_{i-1})e^{-(V_{a,i}-V_{a,i-1})}
	+ a_{i-1}(1-\rho_i)e^{-(V_{a,i-1}-V_{a,i})} \\
	&\quad - a_i (1-\rho_{i+1})e^{-(V_{a,i}-V_{a,i+1})}
	+ a_{i+1}(1-\rho_i)e^{-(V_{a,i+1}-V_{a,i})},
	\\
	2 C h^2 \partial_t r_i &= - r_i (1-\rho_{i-1})e^{-(V_{r,i}-V_{r,i-1})}
	+ r_{i-1}(1-\rho_i)e^{-(V_{r,i-1}-V_{r,i})} \\
	&\quad - r_i (1-\rho_{i+1})e^{-(V_{r,i}-V_{r,i+1})}
	+ r_{i+1}(1-\rho_i)e^{-(V_{r,i+1}-V_{r,i})},
	\end{split}
\end{align}
% with $\tilde{C}= \frac{1}{\tilde{\epsilon}}$ and $C = \frac{1}{2 \overline{\epsilon}}$.
% Note that we divided by $\tilde{a}$ or $\tilde{r}$ respectively.
where 
$$
% \rho_i = \frac{a_i}{a_{\mathrm{max}}} + \frac{r_i}{r_{\mathrm{max}}}
\rho_i = \frac{a_i+r_i}{v_{\mathrm{max}}}
$$
denotes the (relative) sum of antero and \gls{retrograde} \gls{vesicle}s with $v_{\textrm{max}}$ denoting the maximal number of \gls{vesicle}s for a cell of width $h$. Furthermore, $V_{a,i} \coloneqq V_a(x_i)$ and $V_{r,i} \coloneqq V_r(x_i)$ are given potentials with $V_a, V_r \colon \R \rightarrow \R$ evaluated at the midpoint of cell $i$ and $C$ is one over the diffusion constant, see Section \ref{sec:scaling} for details. Roughly speaking, on the right hand sides of the above equations all terms with positive sign correspond to particles that jump into cell $i$ from the neighbouring cells while negative terms remove those that jump out.\medskip

2. \emph{Coupling to \gls{soma} and pools}:
We assume that all lattices are connected to the \gls{soma} at their first lattice site $i=1$. There, we have the following effects:
\begin{itemize}
 \item Retrograde \gls{vesicle}s leave the neurite and enter the \gls{soma} with a rate $\beta_r(\Lambda_{som}) r_1$, if there is enough space, where $\Lambda_{som}$ denotes the number of \gls{vesicle}s currently in the \gls{soma} and $\beta$ is a velocity that depends on this quantity.
 
 \item Anterograde \gls{vesicle}s can leave the \gls{soma} and enter the lattice, if there is enough space, i.e. if $\rho_1 < 1$. In this case, they enter with a given rate $\alpha_a(\Lambda_{som})(1-\rho_N)$, that also depends on the number of \gls{vesicle}s in the \gls{soma}.
\end{itemize}
At site $N$, the neurites are connected to their respective pools (\gls{growth cone}s) and we have that:
\begin{itemize}
 \item Anterograde \gls{vesicle}s leave the lattice and enter the pool with rate $\beta_a(\Lambda_{N}) r_N$, where again the velocity $\beta$ depends on the number of particles in the pool.
 \item Retrograde particles move from the pool into the lattice with rate $\alpha_r(\Lambda_{N})$, again only if space on the lattice is available. This yields the effective rate $\alpha_r(\Lambda_{N})(1-\rho_1)$.
\end{itemize}
Since we assume that both the pool and the \gls{soma} have a maximal capacity that cannot be exceeded, we make the following choices for in- and out-flux rates 
\begin{align*}
	\alpha_q(\Lambda_j) = \alpha_q \frac{\Lambda_j}{\Lambda_j^{\max}} \text{ and } \beta_q(\Lambda_j) = \beta_q ( 1 - \frac{\Lambda_j}{\Lambda_j^{\max}}),\quad k \in \{a,r\},\; q \in \{\mathrm{som},\mathrm{N}\}.
\end{align*}
This yields the following equations at the tips and the \gls{soma}:
\begin{align}\label{eq:discrete_boundary}
	\begin{split}
	2  C  h^2 \partial_t a_1 
	&=
	-  a_1 (1-\rho_{2})e^{-C(V_{a,1}-V_{a,2})}
	+  a_2(1-\rho_1)e^{-C(V_{a,2}-V_{a,1})} 
	\\
	&\qquad + C h\alpha_a (1-\rho_1), 
	\\
	2 C h^2 \partial_t a_N 
	&= -  a_N (1-\rho_{N-1})e^{-C(V_{a,N}-V_{a,N-1})}
	+  a_{N-1}(1-\rho_N)e^{-C(V_{a,N-1}-V_{a,N})} 
	\\
	&\qquad - C h\beta_a  a_N ,
	\\
	2  C  h^2 \partial_t r_1 
	&= -  r_1 (1-\rho_2)e^{-C(V_{r,x_1}-V_{r,2})}
	+  r_2(1-\rho_1)e^{-C(V_{r,2}-V_{r,1})} \\
	&\qquad - C h\beta_r r_1,
	\\
	2 C h^2 \partial_t r_N 
	&= -  r_N (1-\rho_{N-1})e^{-C(V_{r,N}-V_{r,N-1})}
	+  r_{N-1}(1-\rho_N)e^{-C(V_{r,N-1}-V_{r,N})} 
	\\
	&\qquad + C h\alpha_r (1-\rho_N). 
	\end{split} 	 
\end{align}
3. \emph{Dynamics in \gls{soma} and pools}: Finally, we have to describe the change of number of \gls{vesicle}s in the \gls{soma} and the respective neurite pools. For now, we assume that the change is only due to \gls{vesicle}s entering and existing, i.e. no growth or degradation terms are included. This yields the following ordinary differential equations
\begin{align}
	\label{eq:poolequations1}
	\partial_t \Lambda_{\text{N}}
	&= 
	- \alpha_{r} \big(1-\rho_{N}\big) \frac{\Lambda_{\text{N}}}{\Lambda_{\text{N}}^{\max}} 
	+ \beta_{a} a_{N} \big(1 - \frac{\Lambda_{\text{N}}}{\Lambda_{\text{N}}^{\max}} \big),
	\\ \label{eq:poolequations2}
	\partial_t \Lambda_{\text{som}}
	&= 
	\beta_{r} r_{1} \big(1 - \frac{\Lambda_{\text{som}}}{\Lambda_{\text{som}}^{\max}}\big) 
	- \alpha_a \big(1-\rho_{1}\big) \frac{\Lambda_{\text{som}}}{\Lambda_{\text{som}}^{\max}}.
\end{align}
4. \emph{Extension to multiple neurites}:
In the case of $M$ neurites, we denote by $a_{i,l}$ and $r_{i,l}$ the concentration of retro- and \gls{anterograde} \gls{vesicle}s in neurite $l$, $l=1,\ldots M$ at site $i$. The pools are then called $\Lambda_{N,l}$ and we also change the names of all parameters accordingly, i.e. we have $\alpha_{r,l}$, $\beta_{r,l}$, $\ldots$. While the equations for the dynamics inside the neurites \eqref{eq:discrete_interior}, at the tip \eqref{eq:discrete_boundary} and for the respective growth cones \eqref{eq:poolequations1} remain unchanged (despite the different notation for the constants), the equation for the soma becomes
\begin{align}
	\partial_t \Lambda_{\text{som}}
&= \sum_{l=1}^M \left[ \beta_{r,l} r_{1,l} \big(1 - \frac{\Lambda_{\text{som}}}{\Lambda_{\text{som}}^{\max}}\big) 
- \alpha_{a,l} \big(1-\rho_{1,l}\big) \frac{\Lambda_{\text{som}}}{\Lambda_{\text{som}}^{\max}}\right].
\end{align}

%A summary of the model in case of a neuron with two neurites can be found in Figure \eqref{fig:2BScetchModelNeuronLattice}. Note that the full model is also stated in the next section.

\begin{remark}[On the Modelling]
	\begin{itemize}
	\item[a)] Note that $\alpha$ has a different physical interpretation than $\beta$.
	Whereas $\alpha$ is given in $\frac{\text{ves}}{\text{sec}}$ and specifies an influx rate, $\beta$ is given in $\frac{\mu\text{m}}{\text{sec}}$ and therefore specifies an outflux velocity.
	This is essential for the boundary contributions in \eqref{eq:discrete_boundary} all having the same physical unit (using that $1-\rho_i$ is already scaled).
	\item[b)] We are not dealing with the domain and the actual concentrations in the pool explicitly but only model the total number of \gls{vesicle}s present. In particualar there is no diffusion or transport in the pools. Instead, we assume that the dynamics are fast compared to those of the neurites. In that way we allow for \gls{vesicle}s that have left one neurite and entered a pool to immediately leave into another neurite.
	\end{itemize}
\end{remark}
\begin{remark}[Coupling]
	Even though the equations in \eqref{eq:discrete_interior} describe the evolution of concentrations, the pools in \eqref{eq:poolequations1}--\eqref{eq:poolequations2} have the physical unit mass. Their coupling naturally arises using the flux as a linking element. Indeed fluxes have the physical unit $\frac{\text{ves}}{\text{sec}}$ as have the terms on the right hand side of \eqref{eq:discrete_boundary} that correspond to the boundary flow as well as the reaction terms for the time evolution of the pools in \eqref{eq:poolequations1}--\eqref{eq:poolequations2}.
\end{remark}
\begin{remark}[Numerical Simulations]\label{rem:numerical} One advantage of our model is that it immediately yields a discretisation for numerical simulations. Indeed, by construction it is already discrete in space and by applying an explicit Euler discretisation we arrive at a fully discrete scheme. This scheme will be used to perform simulations in Section \ref{sec:numerics}. There, we will also present some of the scheme's properties. 
\end{remark}

\subsection{Scaling}\label{sec:scaling}

Next we transform all relevant variables into an appropriate scaled and dimensionless form, where we always indicate the corresponding dimensionless form with a bar and the typical size with a tilde. Thus e.g. $\overline{r} = \frac{r}{\tilde{r}}$ is a dimensionless quantity. Note that we will then omit the bar everywhere after this section for reasons of readability.

Motivated by the discussion in Section \ref{sec:biological} we make the following choices: The \textit{typical length} is $\tilde{L} = 50 ~\mu \text{m}$, the \textit{typical time} is $\tilde{t} = 100 ~\text{sec}$, the \textit{typical diffusion constant} is $\tilde{\epsilon} = 10^{-1}~ \frac{\mu\text{m}^2}{\text{sec}}$, the \textit{typical potential} is $\tilde{V} = 1 ~\frac{\mu\text{m}^2}{\text{sec}}$. 
The \textit{typical influx} is $\tilde{\alpha} = 1~ \frac{\text{vesicles}}{\text{sec}}$ and the \textit{typical outflow velocity} is $\tilde{\beta} = 10^{-1}~ \frac{\mu\text{m}}{\text{sec}}$, thus the different boundary conditions have the same unit of measurement.
As the \textit{typical diameter} of one \gls{vesicle} is 130 nm and the neurite diameter is 1 $\mu$m, the maximal density is given by $\rho_{\text{max}} =  \frac{8 ~\text{vesicles}}{0,13 ~\mu \text{m}} \approx 60 \frac{\text{vesicles}}{\mu \text{m}}$.
The \textit{typical density} of \gls{anterograde} and \gls{retrograde} particles is $\tilde{a}, \tilde{r} = 15~ \frac{\text{vesicles}}{\mu \text{m}}$, which corresponds to a half filled neurite.

As $1- \rho$ is already scaled, the equations \eqref{eq:discrete_interior} transform to, for $i=2,\ldots, N-1$,
\begin{align}\label{eq:discrete_interior_scaled}
	\begin{split}
	\frac{1}{\lambda_\epsilon} \bar C h^2 \partial_t a_i 
	&= - a_i (1-\rho_{i-1})e^{-\overline{C}\tilde{C}\tilde{V}(\overline{V}_{a,i}-\overline{V}_{a,i-1})}	
	+ a_{i-1}(1-\rho_i)e^{-\overline{C}\tilde{C}\tilde{V}(\overline{V}_{a,i-1}-\overline{V}_{a,i})} 
	\\
	&\quad- a_i (1-\rho_{i+1})e^{-\overline{C}\tilde{C}\tilde{V}(\overline{V}_{a,i}-\overline{V}_{a,i+1})}
	+ a_{i+1}(1-\rho_i)e^{-\overline{C}\tilde{C}\tilde{V}(\overline{V}_{a,i+1}-\overline{V}_{a,i})},
	\\
	\frac{1}{\lambda_\epsilon} \bar  C h^2 \partial_t r_i 
	&= - r_i (1-\rho_{i-1})e^{-\overline{C}\tilde{C}\tilde{V}(\overline{V}_{r,i}-\overline{V}_{r,i-1})}
	+ r_{i-1}(1-\rho_i)e^{-\overline{C}\tilde{C}\tilde{V}(\overline{V}_{r,i-1}-\overline{V}_{r,i})} 
	\\
	&\quad- r_i (1-\rho_{i+1})e^{-\overline{C}\tilde{C}\tilde{V}(\overline{V}_{r,i}-\overline{V}_{r,i+1})}
	+ r_{i+1}(1-\rho_i)e^{-\overline{C}\tilde{C}\tilde{V}(\overline{V}_{r,i+1}-\overline{V}_{r,i})},
	\end{split}
\end{align}
with $\tilde{C}= \frac{1}{\tilde{\epsilon}}$ and $\bar C = \frac{1}{2 \overline{\epsilon}}$.
Thus the product of all typical variables appearing in the summands of the previous two equations are\begin{align}
	\lambda_\epsilon = \frac{\tilde{t}}{2\tilde{L}^2} \tilde{\epsilon} = \frac{100~ \text{sec}}{2500 \mu\text{m}^2} 10 ^{-1} \frac{\mu\text{m}^2}{\text{sec}} 
	= 4\cdot 10^{-3} \text{ and }
	\lambda_V = \frac{\tilde{t}}{2\tilde{L}^2} \tilde{V} = 0.04
	\label{2AFactorsThatOccurByScaling}
\end{align}
after cancellation of $\overline{a}$ and $\overline{r}$ respectively on both sides.
Note that the scaling parameters for the boundary conditions can be calculated by multiplying the boundary conditions with $\frac{\tilde{L}}{\tilde{t}}$ and additionally scaling terms corresponding to in- and outflux with $\tilde{\gamma}$. 
We obtain
 \begin{align}\label{eq:discrete_boundary_scaled}
	\begin{split}
	\bar C  h^2 \partial_t a_1 
	&=
	- \lambda_\epsilon a_0 (1-\rho_{2})e^{-\bar C\tilde{C}\tilde{V}(\overline{V}_{a,1}-\overline{V}_{a,2})}
	+ \lambda_\epsilon a_2(1-\rho_1)e^{-\bar C\tilde{C}\tilde{V}(\overline{V}_{a,2}-\overline{V}_{a,1})} 
	\\
	&\qquad + \lambda_{in} \bar C h\alpha_a (1-\rho_1), 
	\\
	\bar C h^2 \partial_t a_N 
	&= - \lambda_\epsilon a_N (1-\rho_{N-1})e^{-\bar C\tilde{C}\tilde{V}(\overline{V}_{a,N}-\overline{V}_{a,N-1})}
	+ \lambda_\epsilon a_{N-1}(1-\rho_N)e^{-\bar C\tilde{C}\tilde{V}(\overline{V}_{a,N-1}-\overline{V}_{a,N})} 
	\\
	&\qquad
	- \lambda_{out} \bar C h\beta_a  a_N,
	\\
	\bar C  h^2 \partial_t r_1 
	&= - \lambda_\epsilon r_1 (1-\rho_2)e^{-\bar C\tilde{C}\tilde{V}(\overline{V}_{r,1}-\overline{V}_{r,2})}
	+ \lambda_\epsilon r_2(1-\rho_1)e^{-\bar C\tilde{C}\tilde{V}(\overline{V}_{r,2}-\overline{V}_{r,1})} 
	\\
	&\qquad - \lambda_{out} \bar C h\beta_r r_1,
	\\
	\bar C h^2 \partial_t r_N 
	&= - \lambda_\epsilon r_N (1-\rho_{N-1})e^{-\bar C\tilde{C}\tilde{V}(\overline{V}_{r,N}-\overline{V}_{r,N-1})}
	+ \lambda_\epsilon r_{N-1}(1-\rho_N)e^{-\bar C\tilde{C}\tilde{V}(\overline{V}_{r,N-1}-\overline{V}_{r,N})} 
	\\
	&\qquad
	+ \lambda_{in} \bar C h\alpha_r (1-\rho_N),
	\end{split} 	 
\end{align}
where we introduced the dimensionless scaling parameters
\begin{align}
	\lambda_{in}
	= \frac{\tilde{t}\tilde{\alpha}}{2\tilde{L}\tilde{a}} 
	= \frac{100~ \text{sec} \cdot 1 \frac{\text{ves}}{\text{sec}}}{ 50~ \mu \text{m} \cdot 15 \frac{\text{ves}}{\mu \text{m}}}
	= 0.1333,
	\qquad
	\lambda_{out}
	= \frac{\tilde{t}\tilde{\beta}}{2\tilde L}
	= \frac{100~ \text{sec} \cdot 10^{-1} \frac{\mu\text{m}}{\text{sec}}}{50 ~ \mu \text{m}}
	= 0.2.
	\label{def:2BScalingParameterBoundary}
\end{align}
Furthermore equation \eqref{2B:ODEPoolConcentrations} that describes the pool concentration requires scaling. Applying the same time scale as above and the same scaling of \gls{vesicle} concentrations yields
\begin{align*}
	\frac{1}{\tilde{t}} \partial_{\overline{t}} ( \tilde\Lambda_{\text{N}}\bar\Lambda_{\text{N}})
	&=
% 	 \frac{\tilde{t}\tilde{\gamma}}{\tilde{L}\tilde{a}}  
	 \Big[
	- \tilde \alpha \bar \alpha_{r} \big(1-\rho_{N}\big) \frac{\Lambda_{\text{N}}}{\Lambda_{\text{N}}^{\max}}
	+\tilde \beta_{a}\bar \beta_{a} \tilde a_{N}\bar a_{N} \big(1 - \frac{\Lambda_{\text{N}}}{\Lambda_{\text{N}}^{\max}} \big) \Big].
\end{align*}
Multiplying by $\tilde t$ and dividing by $\tilde \Lambda_{N}$ gives
\begin{align*}
     \partial_{\overline{t}}  \bar\Lambda_{\text{N}}
	&=
% 	 \frac{\tilde{t}\tilde{\gamma}}{\tilde{L}\tilde{a}}  
	 \Big[
	- \frac{\tilde \alpha \tilde{t}}{\tilde\Lambda_{\text{N}}}\bar \alpha_{rN} \big(1-\rho_{N}\big) \frac{\Lambda_{\text{N}}}{\Lambda_{\text{N}}^{\max}}
	+\frac{\tilde t\tilde {\beta}_{aN}\tilde a_{N}}{\tilde\Lambda_{\text{N}}}\bar \beta_{aN} \bar a_{N} \big(1 - \frac{\Lambda_{\text{N}}}{\Lambda_{\text{N}}^{\max}} \big) \Big].
\end{align*}
Choosing $\tilde \Lambda_{N} = 2\tilde L\tilde a_{N} = 50~ \mu \text{m} \cdot 15 ~\frac{\text{vesicles}}{\mu\text{m}} = 750 ~\text{vesicles}$, we finally arrive at
\begin{align*}
     \partial_{\overline{t}}  \bar\Lambda_{\text{N}}
	&=
% 	 \frac{\tilde{t}\tilde{\gamma}}{\tilde{L}\tilde{a}}  
	 \Big[
	- \lambda_{in}  \bar\alpha_{r}  \big(1-\rho_{N}\big) \frac{\Lambda_{\text{N}}}{\Lambda_{\text{N}}^{\max}}
	+\lambda_{out}\bar \beta_{aN} \bar a_{N} \big(1 - \frac{\Lambda_{\text{N}}}{\Lambda_{\text{N}}^{\max}} \big) \Big]
\end{align*}
and
\begin{align*}
	\partial_{\overline{t}}  \bar\Lambda_{\text{som}}
	&= 
	\lambda_{out} \bar\beta_{r} \bar r_{1} \big(1 - \frac{\Lambda_{\text{som}}}{\Lambda_{\text{som}}^{\max}}\big)  
	- \lambda_{in} \bar\alpha_{a} \big(1-\rho_{1}\big) \frac{\Lambda_{\text{som}}}{\Lambda_{\text{som}}^{\max}}.
% 		\label{eq:2BPoolODEscaled}
\end{align*}
Again, the generalization to more than one neurite is straight forward.
\begin{remark}[Choice of Typical Parameters]
	The identification $\tilde \Lambda_{N} = 2\tilde L\tilde a_{N}$ in the paragraph above is natural as there is a proportion between the size of the pools and the size of the neurites in reality, where the prefactor 2 corresponds to the fact that we deal with two types of species.
	This proportion should be reflected in the typical value $\tilde{\Lambda}$.
\end{remark}

\section{Macroscopic Cross Diffusion in a Model Neuron with Pools}\label{sec:ModelNeuronPools}
It is well known that lattice models as the one described in the preceeding section have a (formal) correspondence to (systems of) partial differential equations \cite{simpson.hughes.ea:diffusion,BurgerFrancescoPietSchlakeNonlinearCrossdiffusion}. Let us briefly summarize the procedure for a single neurite: First we chose $h = 1/N$ so that the lattice has exactly length one and fix the continuous domain $\Omega = [0,1]$. Now for each lattice cell denote by $x_i \in \Omega$ its midpoint and assume the existence of smooth functions $r=r(x,t)$ and $a=a(x,t)$ such that $r_i(t) = r(x_i,t)$ and $a_i(t)=a(x_i,t)$. With this notation, we can formally apply Taylor's formula to the right hand sides of equations \eqref{eq:discrete_interior_scaled}, up to second order. For example, for the first equation in \eqref{eq:discrete_interior_scaled}, this yields
\begin{align*}
	C h^2 \partial_t a(x_i,t) &= 
	\lambda_\epsilon \Big(
	\frac{1}{2}h^2\left( a(x_i,t) \partial_{xx} \rho(x_i,t) + \partial_{xx}a(x_i,t) (1-\rho(x_i,t)\right) \Big)
	\\
&+ \lambda_V \Big( Ch^2 \Big[ a(x_i,t)\partial_x \rho(x_i,t) \partial_x V_a (x_i,t) - \partial_x a(x_i,t)(1-\rho(x_i,t)) \partial_x V_a(x_i,t).
\\
&- a(x_i,t)(1-\rho(x_i,t)) \partial_{xx} V_a(x_i,t) \Big] + O(h^3) \Big),
\end{align*}
where $O(h^3)$ refers to remaining terms of order $h^3$. Then we divide both sides by $h^2$ and take the limit $h\to 0$ which yields
\begin{align}\label{eq:macro}
\begin{split}
 \partial_t a + \partial_x J_a  &= 0 \text{ with } J_a := -\left(\lambda_\epsilon[(1-\rho))\partial_x a + a\partial_x \rho] + \lambda_V a(1-\rho))\partial_x V_a\right), \\
  \partial_t r + \partial_x J_r  &= 0 \text{ with } J_r := -\left(\lambda_\epsilon[(1-\rho))\partial_x r + r\partial_x \rho] + \lambda_V r(1-\rho))\partial_x V_r\right),
 \end{split}
\end{align}
on $\Omega \times (0,T)$, having applied the same procedure to retrograde vesicles. Equations \eqref{eq:discrete_boundary} results in the boundary conditions
\begin{equation}\label{eq:macro_cross_diff_bc} 
\begin{aligned}
- J_{a} \cdot {n_1}  &= \lambda_{in} \alpha_{a} \frac{\Lambda_{\text{som}}}{\Lambda_{\text{som}}^{\max}} (1-\rho) \quad &&\text{ at }x = 0,
\\
J_{a}\cdot {n_2} &= \lambda_{out}\beta_{a} \big(1 - \frac{\Lambda_{\text{N}}}{\Lambda_{\text{N}}^{\max}} \big) a \quad  &&\text{ at } x = 1,
\\
J_{r} \cdot {n_1}  &= \lambda_{out}\beta_{r}\big(1 - \frac{\Lambda_{\text{som}}}{\Lambda_{\text{som}}^{\max}}\big)  r \quad  &&\text{ at }x=0,
\\
-J_{r}\cdot {n_2} &= \lambda_{in} \alpha_{r} \frac{\Lambda_{\text{N}}}{\Lambda_{\text{N}}^{\max}} (1-\rho) \quad &&\text{ at }x = 1,
\end{aligned} 
\end{equation}
where $n_1$ and $n_2$ denote the outward pointing unit vectors at $x=0$ and $x=1$, respectively. In the case of a single neurite with fixed in- and outflow boundary conditions (i.e. without considering the pools explicitly), this has been carried out in detail in \cite{SchlakePhD2011}.
Passing to the limit in the ODEs for the pools yields 
\begin{align}
	\begin{split}
	\label{2B:ODEPoolConcentrations}
	\partial_t \Lambda_{\text{N}}
	&= 
	- \alpha_{r} \big(1-\rho(1)\big) \frac{\Lambda_{\text{N}}}{\Lambda_{\text{N}}^{\max}} 
	+ \beta_{a} a(1) \big(1 - \frac{\Lambda_{\text{N}}}{\Lambda_{\text{N}}^{\max}} )\big),
	\\
	\partial_t \Lambda_{\text{som}}
	&= 
	\beta_{r} r(0) \big(1 - \frac{\Lambda_{\text{som}}}{\Lambda_{\text{som}}^{\max}}\big) 
	- \alpha_{a} \big(1-\rho(0)\big) \frac{\Lambda_{\text{som}}}{\Lambda_{\text{som}}^{\max}}, 
%	\\
%	&\qquad + \beta_{rN2} r_{N2}(0)  \big(1 - \frac{\Lambda_{\text{som}}}{\Lambda_{\text{som}}^{\max}}\big) 
%	- \alpha_{aN2} \big(1-\rho_{N2}(0)\big) \frac{\Lambda_{\text{som}}}{\Lambda_{\text{som}}^{\max}},
	\end{split}
\end{align}
i.e. the only difference in contrast to \eqref{eq:poolequations1}--\eqref{eq:poolequations2} is the fact that the concentrations $a, r$ and $\rho$ are now functions on a continuous domain $\Omega$ instead of a discrete grid. Therefore we wrote $r(0)$ instead of $r_1$, etc. 

% where $i \in \lbrace 1, 2 \rbrace$ and $\rho_{N1}(0), \rho_{N2}(0), \rho_{Ni}(1)$ are well defined as $\rho \in H^1(\Omega) \hookrightarrow C(\R)$.
%
\begin{figure}
\centering
	\includegraphics[width=1\textwidth]{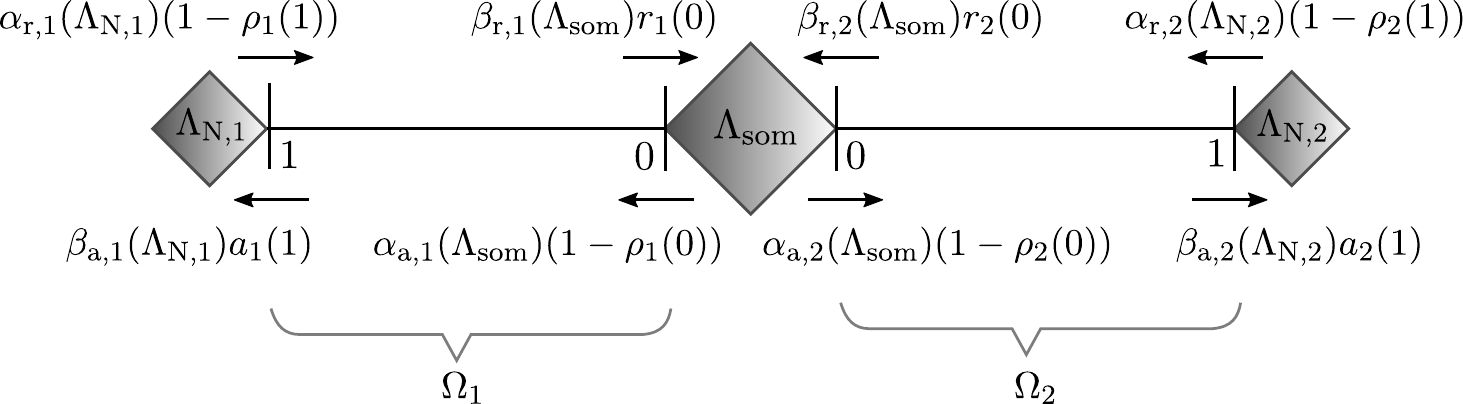}
	\caption{\textbf{Sketch of the model neuron:} The model neuron consists of two neurites and indicated boundary flow in the domain $\Omega = \Omega_1 \cup \Omega_2$, where the two unit intervals $\Omega_1$ and $\Omega_2$ correspond to two neurites.
	The squares correspond to pools where \gls{vesicle}s can be stored, i.e. the pool in the middle corresponds to the \gls{soma} and the pools at the tips of the neurites correspond to the corresponding \gls{growth cone}s. For an easy visualization $\Omega_1$ is illustrated as a mirrored copy of $\Omega_2$.}
	\label{fig:2BScetchModelNeuron}
\end{figure}

In the situation of two neurites, we will have equations \eqref{eq:macro} for each neurite with appropriate boundary condition and again the ODE for $\Lambda_{som}$ will contain as a right hand the sum of all in- and outfluxes. This situation is summarized in Figure \ref{fig:2BScetchModelNeuron}. In particular, we see that formally the total mass of the system is preserved, as expected.
%\begin{align}
%	\begin{split}
% 0 &= \frac{d}{dt}\left( \int_\Omega a\;dx + \int_\Omega r\;dx + \Lambda_{\text{N,1}} + \Lambda_{\text{N,2}}+ \Lambda_{\text{som}} \right)
% \\
% &=
% \int_\Omega \partial_t a \;dx + \int_\Omega \partial_t r \;dx + \partial_t \Lambda_{\text{N,1}} +\partial_t\Lambda_{\text{N,2}}+ \partial_t \Lambda_{\text{som}}
% \\
% &=
% \int_{\partial\Omega} -J_a\cdot n \;ds + \int_{\partial\Omega} -J_r\cdot n \;ds + \partial_t \Lambda_{\text{N,1}} +\partial_t\Lambda_{\text{N,2}}+ \partial_t \Lambda_{\text{som}} .
% \label{eq:massPreserve}
% \end{split}
%\end{align}

% \begin{remark}
% The maximal mass $M_{\max}$ is 
% \begin{align*}
% 	M_{\max} &= \vert \Omega_1 \vert + \vert \Omega_2 \vert + \Lambda_{\text{som}}^{\max} + \Lambda_{\text{N,1}}^{\max} + \Lambda_{\text{N,2}}^{\max},
% \end{align*}
% where $\vert \Omega \vert$ corresponds to the length of the domains.
% \end{remark}

\begin{remark}[Analysis of the Model] \label{rem:Analysis}
 	The focus of this paper is to gain an understanding of the distribution of \gls{vesicle}s during the growth of neurites based on the discrete model introduced in Section \ref{sec:modelling} and its numerical simulation. However, from a mathematical point of view it is also very interesting to study the macroscopic counterpart of the model given by the system of equations \eqref{eq:macro_cross_diff_bc}--\eqref{2B:ODEPoolConcentrations}. We therefore briefly point out the relevant questions and difficulties in the mathematical analysis of this model. 
	
	Clearly, most important is the question of existence and uniqueness of solutions. From an application point of view, also the long time behaviour is relevant. As for existence, a number of results on cross-diffusion equations of type \eqref{eq:macro_cross_diff_bc} is available, \cite{jungel:boundedness-by-entropy*5,burger.francesco.ea:nonlinear,ehrlacher.bakhta:cross-diffusion}, and also the flux boundary conditions \eqref{eq:macro_cross_diff_bc} have been analysed before, \cite{burger.pietschmann:flow*1}. The main issue when applying these results to our model is the following. The present theory shows existence of solutions in the spaces
	\begin{align*}
	r_i,\, a_i &\in L^2((0,T);L^2(\Omega))\cap H^1((0,T);(H^1)^*(\Omega)),\\
	\rho_i &\in L^2((0,T);H^1(\Omega))\cap H^1((0,T);(H^1)^*(\Omega)).
	\end{align*}
	Thus, making use of the embedding of $H^1$ into the space of continuous functions, it makes sense to evaluate $\rho$ at a point of the boundary (e.g. $\rho(0)$). Unfortunately, this regularity is not available for the concentrations $r$ and $a$ so that we cannot evaluate them at the boundary as would be necessary for the boundary conditions \eqref{eq:macro_cross_diff_bc} to be well-defined. Thus one would need an improved regularity theory (which seems out of reach at present) or one needs to modify the model in a way which is consistent with the biological modelling on the discrete level (e.g. by allowing particles to switch places). As for the long time behaviour, the numerical simulations of Section \ref{sec:numerics} suggest that metastable states exist. Their analysis is another interesting problem and we postpone both issues to future work.
\end{remark}

\section{Numerical Simulations}\label{sec:numerics}
In order to derive a fully discrete numerical scheme, see also Remark \ref{rem:numerical}, we use an explicit Euler discretisation for the time derivatives in \eqref{eq:discrete_interior_scaled} and \eqref{eq:discrete_boundary_scaled}. 
Subdividing the interval $[0,T]$ into $K$ intervals we denote by $\tau = T/K$ the step size and by $a_i^k,\,r_i^k$ the respective concentrations at time $t_k = k\tau$. Within the neurites this results in the scheme
\begin{align}\label{eq:num_scheme_bulk}
\begin{split}
	a_i^{k+1} &= a_i^k + \tau H G_a^k,\\
	r_i^{k+1} &= r_i^k + \tau H G_r^k,
	\end{split}
	\qquad i = 2,...N-1,
\end{align}
where 
$H = 2\lambda_\epsilon\frac{\epsilon}{ h^2}$ and 
\begin{align*} 
\begin{split}
	G_q^k = - q_i^k (1-\rho_{i-1}^k)e^{-C\tilde{C}\tilde{V}(V_{q,i}-V_{q,i-1})}
	+ q_{i-1}^k(1-\rho_i^k)e^{-C\tilde{C}\tilde{V}(V_{q,i-1}-V_{q,i})} \\
	- q_i^k(1-\rho_{i+1}^k)e^{-C\tilde{C}\tilde{V}(V_{q,i}-V_{q,i+1})}
	+ q_{i+1}^k(1-\rho_i^k)e^{-C\tilde{C}\tilde{V}(V_{q,i+1}-V_{q,i})}
	\end{split} \quad q \in \{a,r\}.
\end{align*}
The evolution at the boundary follows by discretising the time derivates in \eqref{eq:discrete_boundary_scaled}, e.g. for anterograde vesicles at the soma we obtain
\begin{align}\label{eq:num_scheme_bdry}
	a_1^{k+1} &= a_1^k + \tau \Big( H G^k_{a,1} +  \frac{\lambda_{in}}{h} \alpha_a (1-\rho_1^k) \Big), 
\end{align}
with  
\begin{align*} 
	G^k_{a,1}
	= - a_1^k (1-\rho_{2}^k)e^{-C\tilde{C}\tilde{V}(V_{a,1}-V_{a,2})}
	+ a_2^k(1-\rho_1^k)e^{-C\tilde{C}\tilde{V}(V_{a,2}-V_{a,1})}.
\end{align*}
For the time discretisation of the ODEs \eqref{2B:ODEPoolConcentrations} for the pools and the \gls{soma} we also use an explicit Euler discretisation with the same time step size. As $\Lambda_{N}$, $\Lambda_{som}$ model a mass and to ensure that the total mass remains conserved we multiply the in- and outflux terms by $h$ and finally obtain the evolution of the pool concentrations via
\begin{align*}
	\Lambda_q^{k+1} = \Lambda_q^k + \tau \Big[ \lambda_{in} ~\text{Influx Terms} + \lambda_{out}~ \text{Outflux Terms} \Big],\quad q \in \{N,som\}.
\end{align*}
To further analyse the properties of this scheme, let us define the constants
\begin{align*}
V_{q,max}^- := \max_{2\le i \le n-1} e^{-C\tilde{C}\tilde{V}(V_{q,i}-V_{q,i-1})},
\quad V_{q,max}^+ := \max_{2\le i \le n-1} e^{-C\tilde{C}\tilde{V}(V_{q,i}-V_{q,i+1})}
\end{align*}
as well as 
\begin{align}
V_{max} := \max(V_{a,max}, V_{r,max})\quad \text{ where } \quad V_{k,max} := \max (V_{q,max}^+,V_{q,max}^-).
\end{align}

\begin{lemma}[Preservation of box constraints]
	Assume that the initial concentrations $a_i^0, r_i^0$ for $i=1,\ldots,N$ are non-negative and satisfy the density constraint $a_i^0 + r_i^0 \leq 1$. Then if 
%	$V,\,C,\,h$ and $\tau$ are chosen such that 
	the (CFL-like) condition 
	$$
	(1- 2\tau HV_{max} - \tau\max(2 HV_{max}, Ch\lambda_{out}\max(\beta_a,\beta_r),Ch\lambda_{in}\max(\alpha_a,\alpha_r))) \ge 0
	$$
	holds we also have 
	$$
	0 \le a_i^k,\, r_i^k,\;\;\; a_i^k+r_i^k \le 1\quad \text{ for }\quad k=1,\ldots, M,\; i =1, \ldots, N,
	$$
	with $a_i^k$, $r_i^k$ computed from $a_i^{k-1}$, $r_i^{k-1}$ via \eqref{eq:num_scheme_bulk}--\eqref{eq:num_scheme_bdry}.
	
\end{lemma}
\begin{proof}
We argue by induction and assume that at time $t_k$ the constraints are satisfied. Indeed, according to \eqref{eq:num_scheme_bulk}, we have that for $i=2,\ldots,N-1$
\begin{align*}
a_i^{k+1} &= \left(1 -  \tau H\left[(1-\rho_{i-1}^k)e^{-C\tilde{C}\tilde{V}(V_{a,i}-V_{a,i-1})}
-(1-\rho_{i+1}^k)e^{-C\tilde{C}\tilde{V}(V_{a,i}-V_{a,i+1})}\right]\right) a_i^{N}\\
&+ a_{i-1}^k(1-\rho_i^k)e^{-C\tilde{C}\tilde{V}(V_{a,i-1}-V_{a,i})} 
+ a_{i+1}^k(1-\rho_i^k)e^{-C\tilde{C}\tilde{V}(V_{a,i+1}-V_{a,i})}\\
&\ge \left(1- 2\tau HV_{max}\right) a_i^k,
\end{align*}
where we used that by assumption $(1-\rho_{i+1}^k)$ and $(1-\rho_{i-1}^k)$ are bounded by one and that the last two terms are non-negative. Thus the condition 
\begin{align}\label{eq:cond_stab_1}
(1-2\tau HV_{max}) \ge 0
\end{align} implies $a_i^{k+1} \ge 0$ and an analogous calculation yields the same condition to ensure non-negativity of $r_i^{k+1}$. To show that $(1-\rho_i^{k+1}) \ge 0$ we note that
\begin{align}\label{eq:cond_stab_2}
(1-\rho_i^{k+1}) \ge \left(1-4\tau HV_{max}\right) (1-\rho_i^{k+1}) 
\end{align}
holds. It remains to consider the boundary contributions. In order to preserve positivity when outflow conditions are present (i.e. for $a_n$ and $r_1$) we obtain the condition
\begin{align}\label{eq:cond_stab_3}
(1-2\tau HV_{max} - \tau Ch\lambda_{out}\max(\beta_a,\beta_r)) \ge 0,
\end{align}
while in order to preserve $\rho \le 1$ at inflow parts we have
\begin{align}\label{eq:cond_stab_4}
(1-2\tau HV_{max} - \tau Ch\lambda_{in}\max(\alpha_a,\alpha_r)) \ge 0.
\end{align}
To have \eqref{eq:cond_stab_1}--\eqref{eq:cond_stab_4} satisfied simultaneously finally yields the assumption
\begin{align*}
(1- 2\tau HV_{max} - \tau \max(2HV_{max}, Ch\lambda_{out}\max(\beta_a,\beta_r),Ch\lambda_{in}\max(\alpha_a,\alpha_r))) \ge 0.
\end{align*}
\end{proof}
\begin{remark}[Natural choice of numerical scheme]
	The classical upwind scheme is not applicable in this context as it only considers the particle movement initiated by the drift term. 
	In this context the drift term of species A can push against the drift of species R.
	This aspect is not covered by the classical upwind scheme.
%	If $C \rightarrow \infty$, we suggest that our algorithm converges against the Scharfetter-Gummel algorithm, see ...	
\end{remark}

For the case to two neurites the algorithm was implemented in MATLAB using $400$ grid points in each domain and a time step size of $\tau = 10^{-5}$. See also the pseudo-code  in subsection \ref{Pseudocodes and Computing Time} in the appendix.

\subsection{Choice of Parameters and their Interpretation}
If not stated otherwise, we use the symmetric initial data shown in Table \ref{eq:2BInitialDatum} for the numeric simulations.
The corresponding value to a parameter in physical units can be calculated by multiplying the typical variable with the value used in the numerics (for reasonability of the data see \cite{tsaneva-atanasova_quantifying_2009}, \cite{urbina_spatiotemporal_2018}, \cite{pfenninger_plasma_2009}).

The potentials $V_a(x) = 1.75~x$ and $V_r(x) = -1.5~x$ translate to that fact that particles of type $a$ move \gls{anterograde} and particles of type $r$ move \gls{retrograde} with different velocities.
In Figure \ref{fig:BioFigureEF} b) the velocity of \gls{vesicle}s that are marked by \gls{Vamp2} is shown.
In practice different species of \gls{vesicle}s with different velocities have been observed (\cite{gumy_map2_2017}, \cite{schlager_bicaudal_2014}). 
As the range of the velocities of the \gls{anterograde} transport is $1 - 2.5$ $\mu \text{m}/ \text{sec}$ and the range of the velocity of \gls{retrograde} transport is  $1 - 2$ $\mu \text{m}/ \text{sec}$ and the typical velocity is 1 $\mu \text{m}/ \text{sec}$, we chose the mean of those ranges.
The diffusion constant $\epsilon$ is not biological meaningful as \gls{vesicle}s do not diffuse but the formal inclusion of this effects justifies to neglect reverse movement of \gls{vesicle}s.
The value of $\epsilon = 0.05$ is purely estimated.

For the maximal density in the pools we do not have a choice:
A neurite of $ 50 ~\mu$m length has volume $V_N = \pi r^2 h = 39,27~ \mu \text{m}^3$.
A \gls{vesicle} with diameter $d= 130$ nm has volume $V_v = \frac{4}{3} \pi r^3 = 0,00115035 ~ \mu \text{m}^3$.
Therefore, a neurite with a length of 50 $\mu$m can contain a maximum of 34 000 \gls{vesicle}s. The \gls{soma} is estimated to contain about 6000 \gls{vesicle}s and the pools in the \gls{growth cone}s at the tip of the neurites about 100 \gls{vesicle}s. 
Consequently, the mass of the \gls{vesicle}s in the \gls{soma} should be 0.175 times as big as the mass of \gls{vesicle}s in a $50~ \mu \text{m}$ long neurite, i.e. $\Lambda_{\text{som}}^{\max}=0.175$ and $\Lambda_{\text{N,1}}^{\max}=\Lambda_{\text{N,2}}^{\max} = 0,0029$.
\begin{table}[h]
	\begin{center}
	\begin{tabular}{c | c | c  | c}
		Variable & Typical Variable & Value in Numerics & Corresponds to \\
		\hline 
		$\Omega_1$ & 50 $\mu$m & [0,1] & [0,50 $\mu$m] \\
		$\Omega_2$ & 50 $\mu$m & [0,3] & [0,15 $\mu$m] \\
		$T$ & 100 s & 50 & 1 h 23 min  \\
		$a_{\text{N,1}}^0, a_{\text{N,2}}^0 $ & 15 $\text{ves}/ \mu \text{m}$ & 0.1 & 1.5 vesicles/$\mu$m  \\
		$r_{\text{N,1}}^0, r_{\text{N,2}}^0$ & 15 $\text{ves} /\mu \text{m}$ & 0.1 & 1.5 vesicles/$\mu$m  \\
		$\epsilon$ & $10^{-1} \mu\text{m}^2 / \text{sec}$ & 0.05 & $5\cdot 10^{-3} \mu \text{m}^2/ \text{sec}$  \\
		$\alpha_{r,1}, \alpha_{r,2},\alpha_{a,1}, \alpha_{a,2}$ & 1 $\text{ves} / \text{sec}$ & 0.8 & 0.8 vesicles/s \\
		$\beta_{r,1}, \beta_{a,2}, \beta_{r,2}, \beta_{a,2}$ & $10^{-1} \mu\text{m}/\text{sec}$ & 15 & 1.5 $\mu$m / s   \\
		$V_a(x)$ & 1 $\mu\text{m}^2/\text{sec}$ & $1.75 ~ x$ & $1.75 ~\mu$m/s \\
		$V_r(x)$ & 1 $\mu\text{m}^2/\text{sec}$ & $-1.5 ~x$ & $-1.5 ~\mu$m/s \\
		$\Lambda_{\text{som}}^0$ & - & 0.12 & 4114 vesicles \\
		$\Lambda_{\text{N,1}}^0, \Lambda_{\text{N,2}}^0$ & - & 0.0015 & 50 vesicles \\
		$\Lambda_{\text{som}}^{\max}$ & - & 0.175 & 6000 vesicles \\
		$\Lambda_{\text{N,1}}^{\max}, \Lambda_{\text{N,2}}^{\max}$ & - & 0.0029 & 100 vesicles 
	\end{tabular}
	\end{center}
	\caption{\textbf{Initial data:} Symmetric initial data used for the numerical simulation and their corresponding real data where the corresponding real data is the product of the typical variable and the value in the numerics by construction.}
	\label{eq:2BInitialDatum}
\end{table}

\subsection{Numeric Results}

In our numerical analyzation we performed two experiments with symmetric initial data for both neurites except for their length, see Table \ref{eq:2BInitialDatum}. 
In the first experiment the domains had a similar length, i.e. $\Omega_1 = [0,1], ~ \Omega_2 = [0, 0.9]$.
In the second experiment there was a significant difference in the length of the two domains, i.e. $\Omega_1 = [0,1], ~ \Omega_2 = [0, 0.3]$.

The results are shown in Figure \ref{fig:2BTimeEvolutionSimilarLength} and \ref{fig:2BTimeEvolutionDifferentLength}. 
Each component of the figures consists of three bars and two diagrams.
The left bar displays the total concentration of \gls{vesicle}s in $\Lambda_{\text{N,1}}$, the bar in the middle the value of $\Lambda_{\text{som}}$ and the right one the value of $\Lambda_{\text{N,2}}$.
The left diagram shows the current \gls{vesicle} concentration of \gls{anterograde} (green) and \gls{retrograde} (red) moving \gls{vesicle}s in neurite 1 and the right the one the concentration in neurite 2.

The aim of these experiments was the analysis of asymmetry-formation arising upon the basis of symmetric initial data.
In particular, we regard asymmetries in the \gls{vesicle} concentration in the \gls{growth cone}s as different growing potentials of the neurites.
Therefore, if the concentration in the pools $\Lambda_{\text{N,1}}$ and $\Lambda_{\text{N,2}}$ are not equal, the neurite with the higher concentration in the pool has a higher growing potential.
This can be justified by the fact that we assume that \gls{vesicle} merging with the membrane at the \gls{growth cone} drive the growing process of the neurite.

As visible in Figure \ref{fig:2BTimeEvolutionDifferentLength} (b) the length difference of the neurites leads to a quick symmetry breaking especially in the \gls{growth cone}s $\Lambda_{\text{N,1}}$ and $\Lambda_{\text{N,2}}$. 
Our numerical experiments even suggest that for a small initial length difference, there is nearly no asymmetry, see Figure \ref{fig:2BTimeEvolutionSimilarLength}.

\begin{figure}
	\begin{center}
		\begin{tabular}{cc}
(a) & \includegraphics[width=0.825\textwidth]{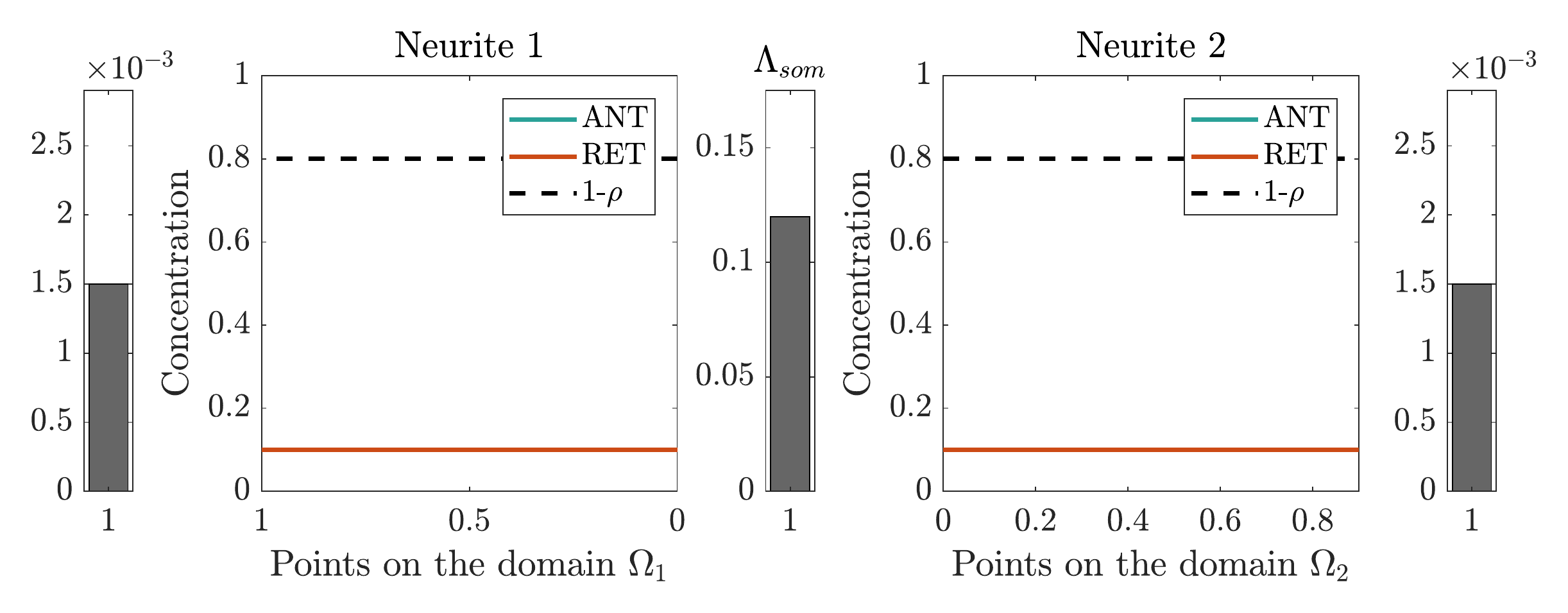} \\ 
(b) & \includegraphics[width=0.825\textwidth]{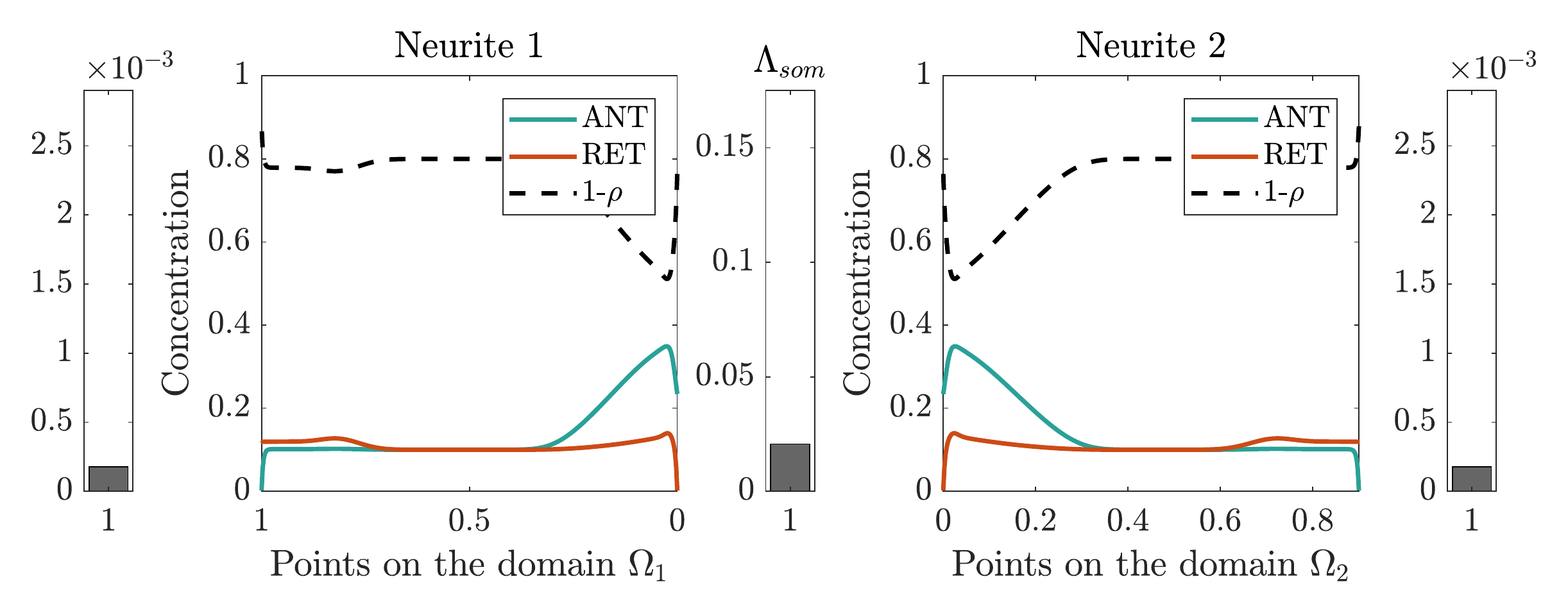} \\
(c) & \includegraphics[width=0.825\textwidth]{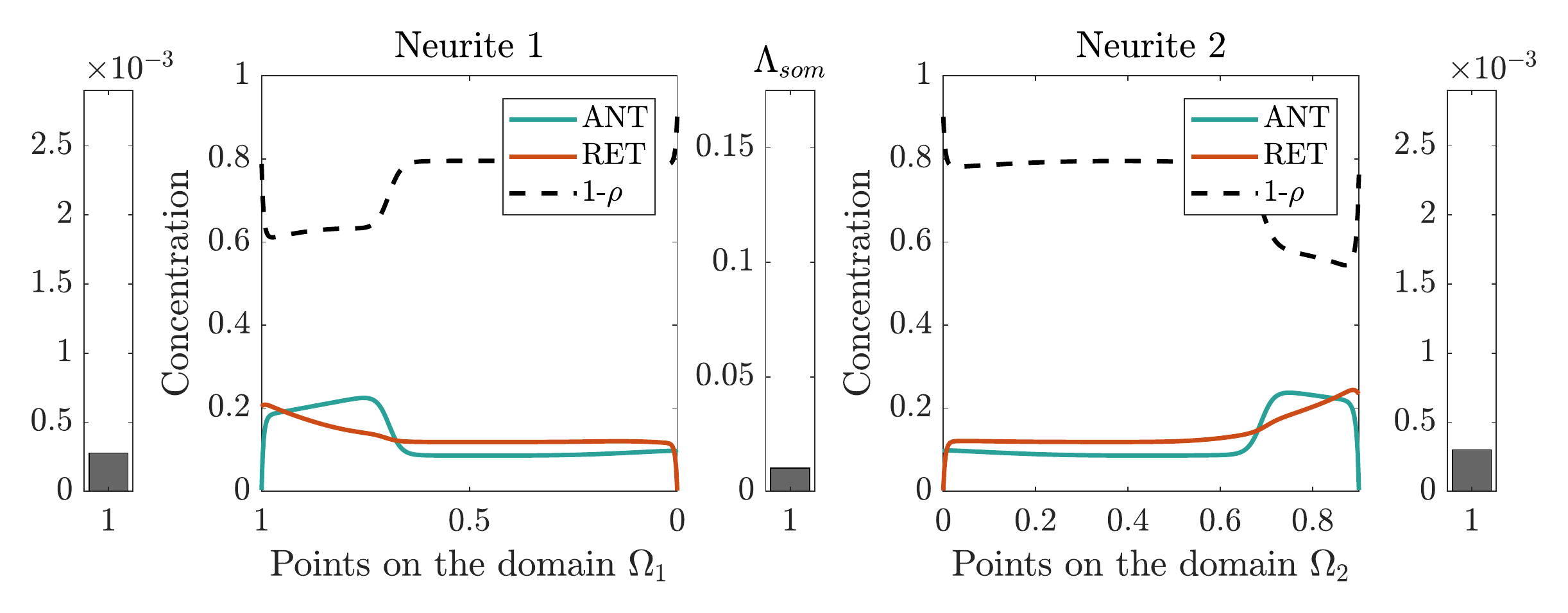} \\ 
(d) & \includegraphics[width=0.825\textwidth]{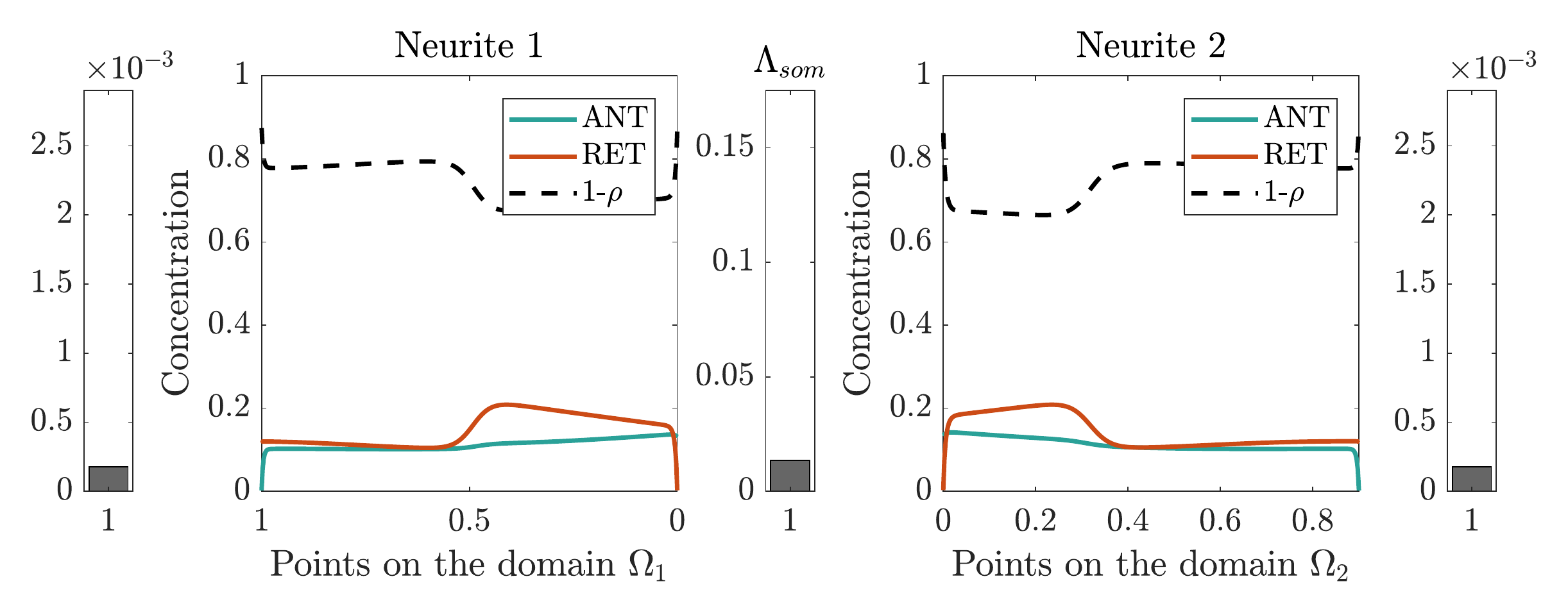} \\ 
		\end{tabular}
	\end{center}
	\caption{\textbf{Evolution over Time of the Vesicle Concentration in two Neurites with Nearly Similar Length:} The vesicle concentration for $\Omega_1 = [0,1], ~ \Omega_2 = [0, 0.9]$ and initial data \eqref{eq:2BInitialDatum}.
	(a) The initial concentration at $t=0$, (b) $t=10$, (c) $t=50$, (d) $T=100$.}
	\label{fig:2BTimeEvolutionSimilarLength}
\end{figure}

\begin{figure}
	\begin{center}
		\begin{tabular}{cc}
(a) & \includegraphics[width=0.825\textwidth]{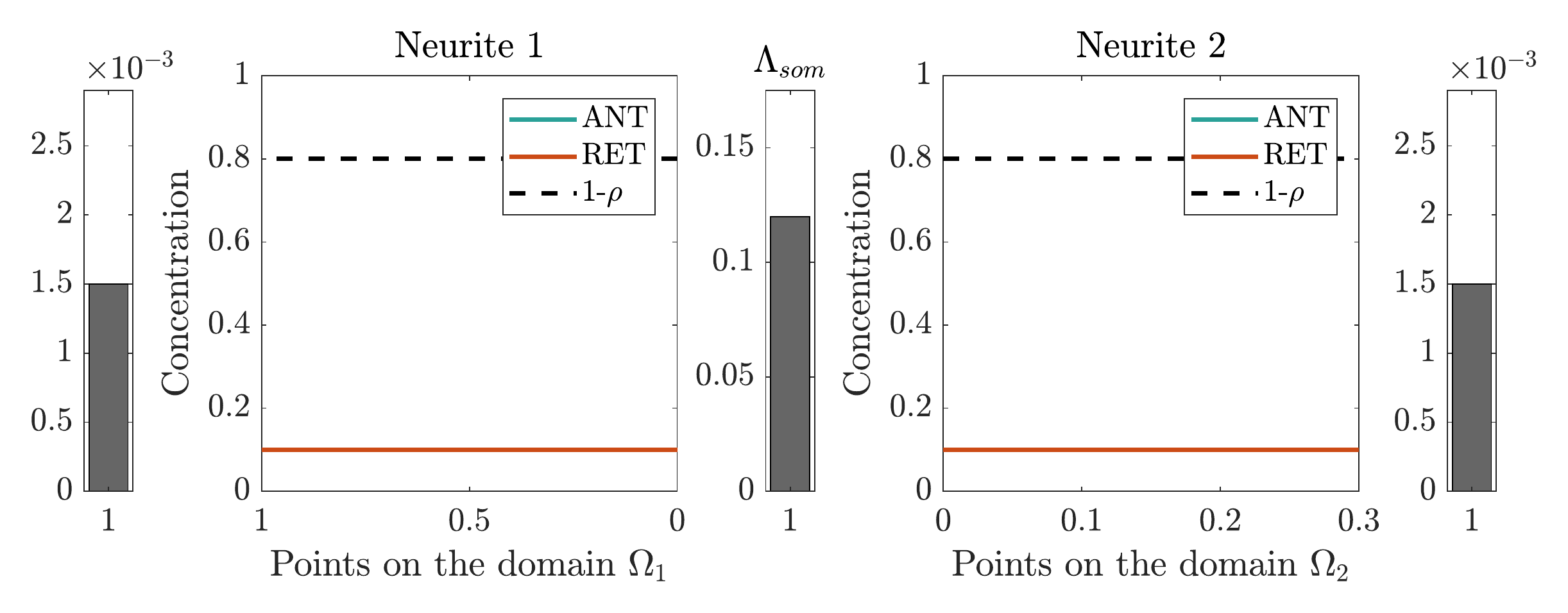} \\ 
(b) & \includegraphics[width=0.825\textwidth]{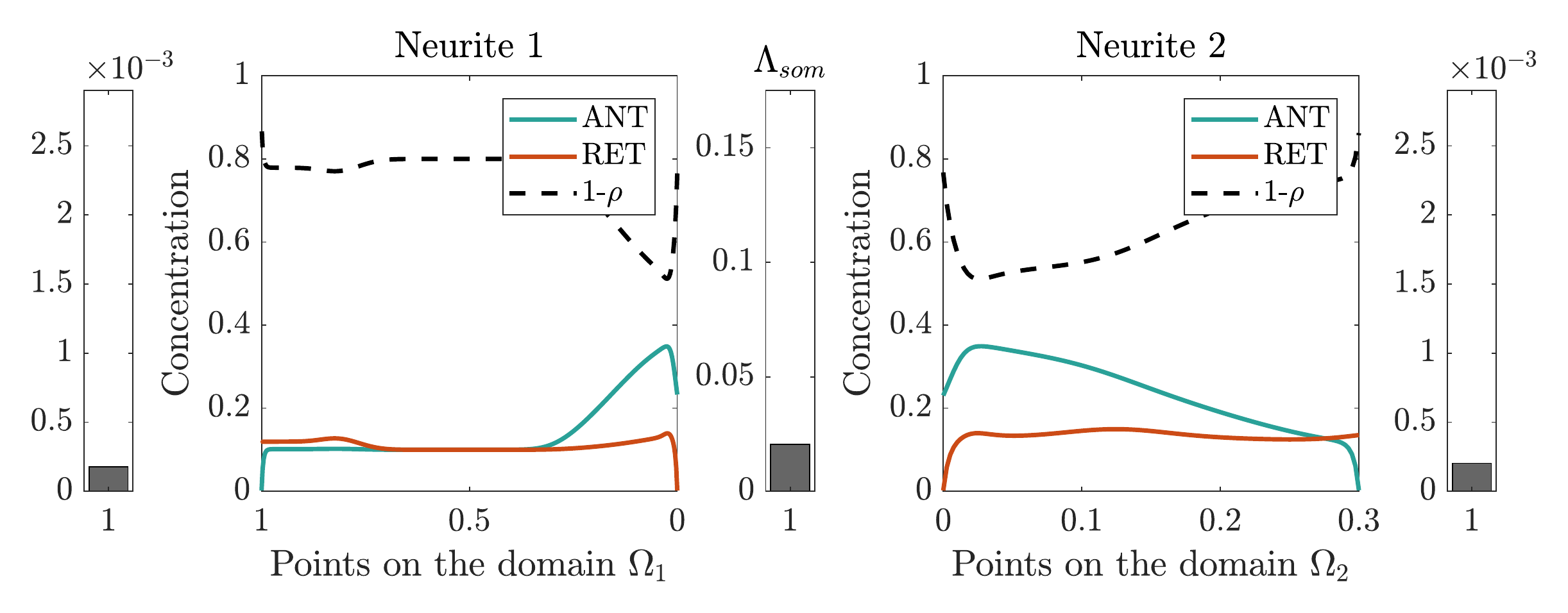} \\
(c) & \includegraphics[width=0.825\textwidth]{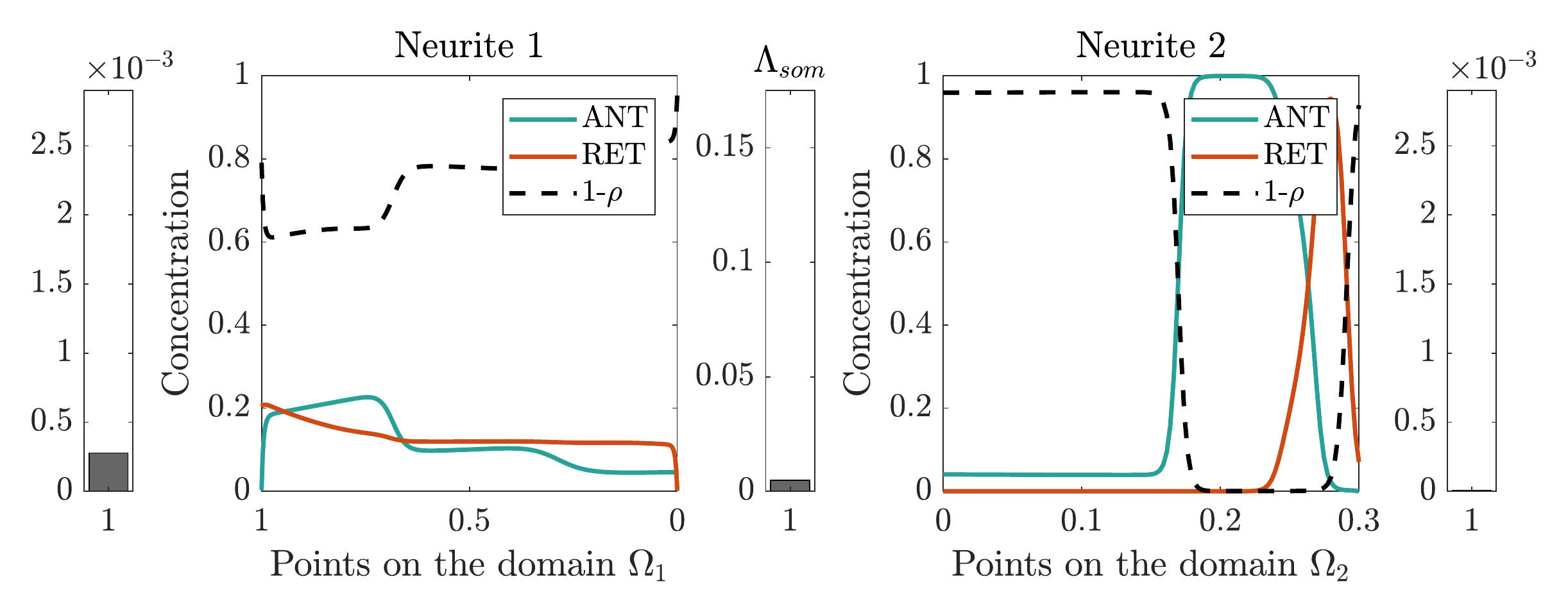} \\ 
(d) & \includegraphics[width=0.825\textwidth]{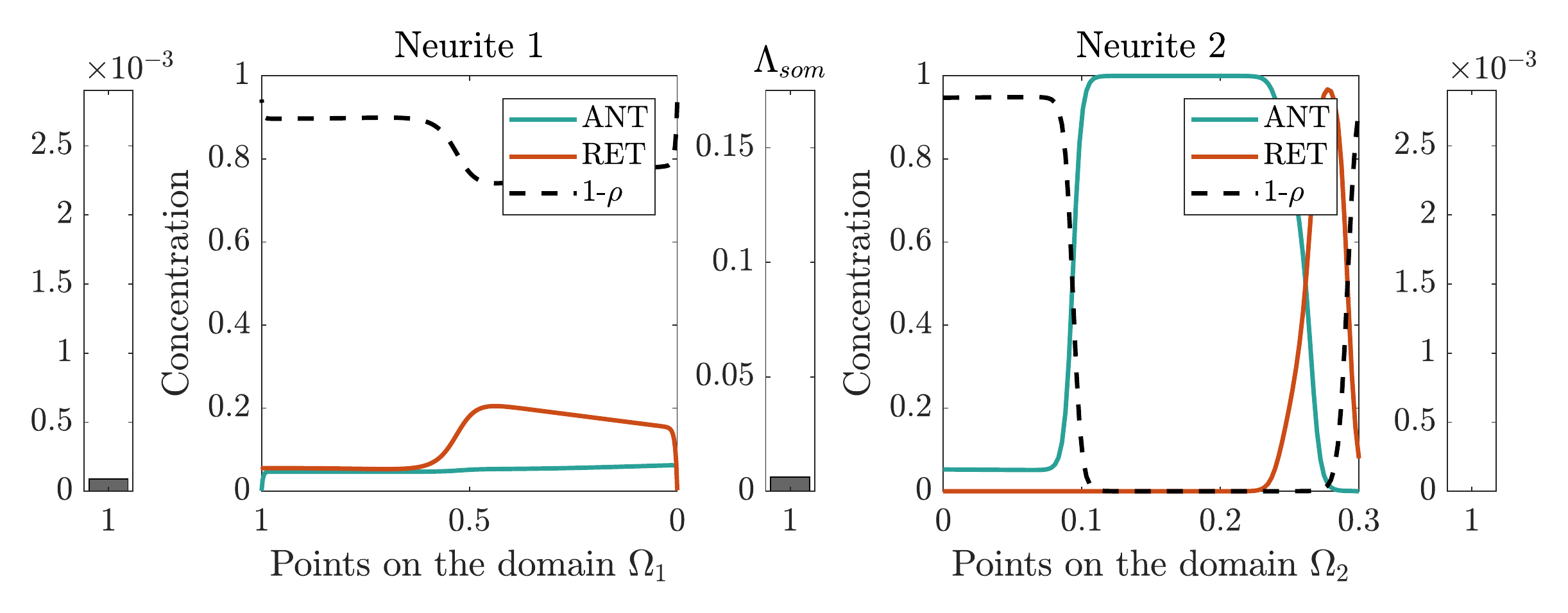} \\ 
		\end{tabular}
	\end{center}
	\caption{\textbf{Evolution over Time of the Particle Concentration in two Neurites with very Different Length:} The vesicle concentration for $\Omega_1 = [0,1], ~ \Omega_2 = [0, 0.3]$ and initial data \eqref{eq:2BInitialDatum}.
	(a) The initial concentration at $t=0$, (b) $t=10$, (c) $t=50$, (d) $T=100$.}
	\label{fig:2BTimeEvolutionDifferentLength}
\end{figure}

\begin{figure}[h]
	\begin{center}
		\begin{tabular}{cc}
			(a) & \includegraphics[width=0.9\textwidth]{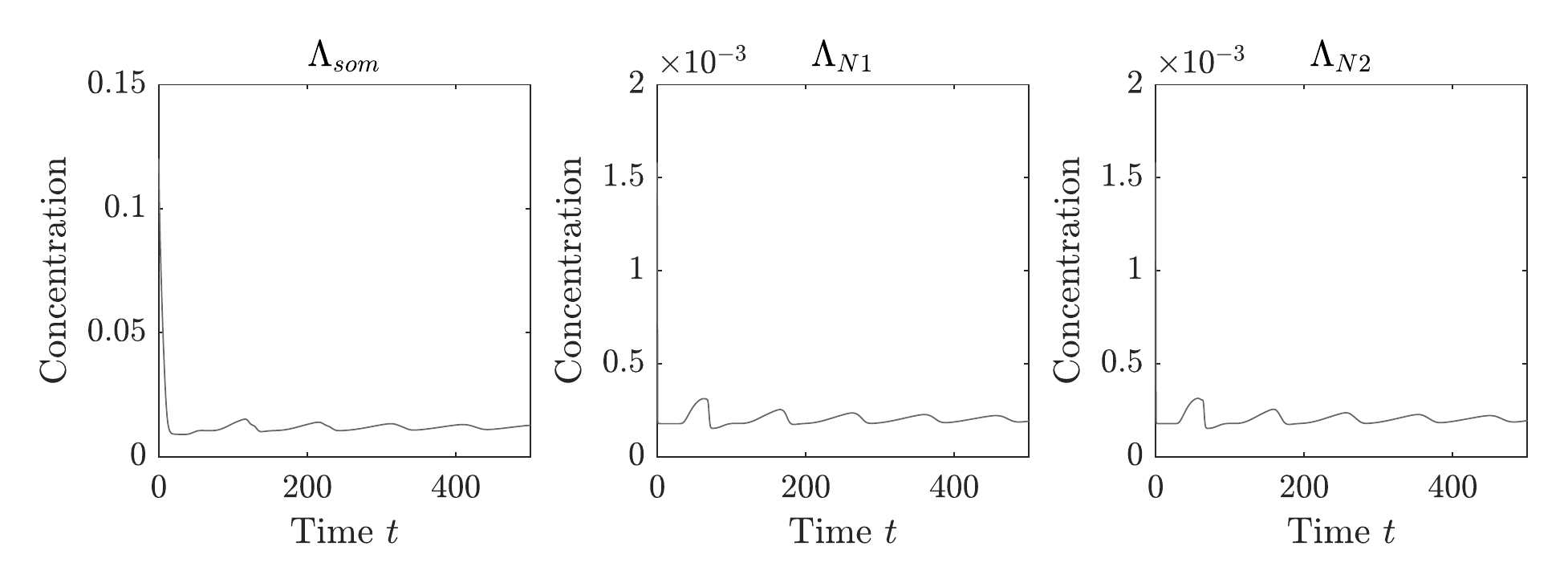} \\
			(b) &
			\includegraphics[width=0.9\textwidth]{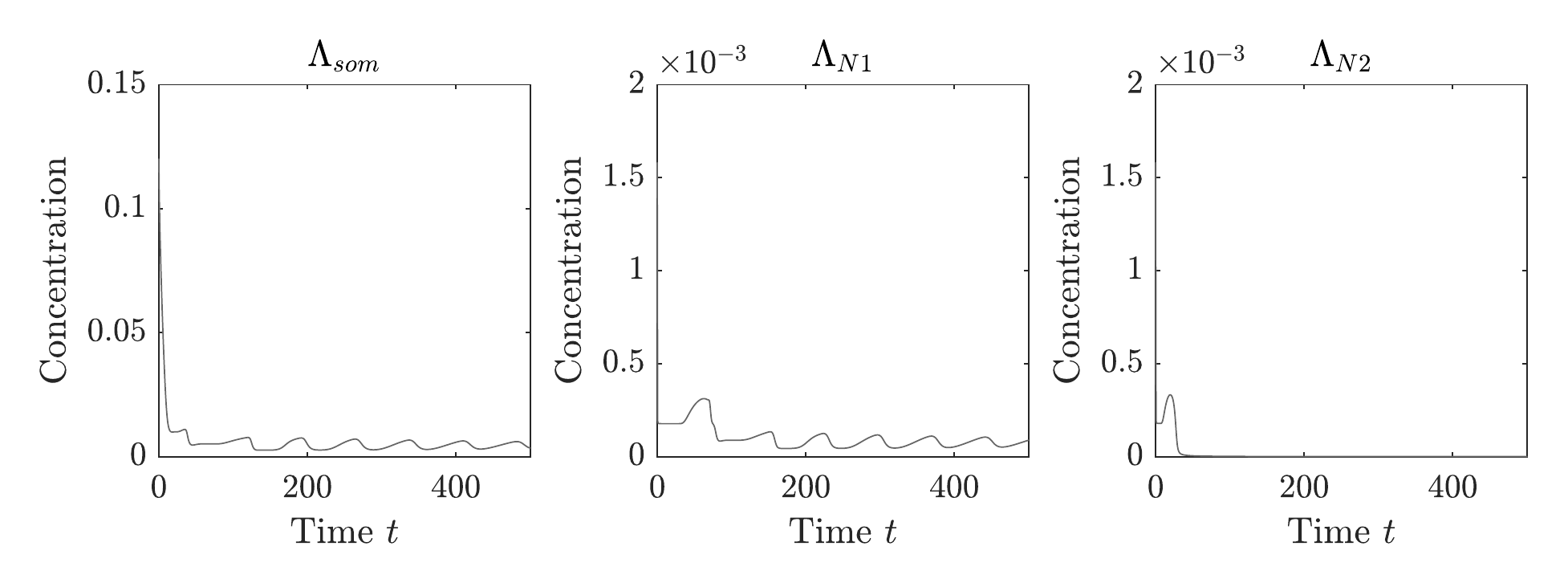}
		\end{tabular}
	\end{center}
	\caption{\textbf{Evolution over time of concentration in the pools} The time evolution of $\Lambda_{\text{som}}, \Lambda_{\text{N,1}}, \Lambda_{\text{N,2}}$ solving \eqref{2B:ODEPoolConcentrations}  and initial data \eqref{eq:2BInitialDatum}. (a) For $\Omega_1 = [0,1], ~ \Omega_2 = [0, 0.9]$, (b) for $\Omega_1 = [0,1], ~ \Omega_2 = [0, 0.3]$.}
	\label{fig:2BPoolConcentration}
\end{figure}

\subsection{Metastability in the Model}

Classically, in the context of dynamic systems a stable state without least energy is called a metastable state. 
Therefore the system stays in that state if no external energy is added whereas a certain amount of energy can result in further time evolution and the system coming to its true stable state with least energy. 

In our model in the case with very different initial lengths we monitor a quite similar but not equal effect where the particle concentration seems to have already converged to its equilibrium until a sudden rapid change in the \gls{vesicle}s concentration occurs at the tip of the longer neurite $(\Lambda_{\text{N,2}})$ after some time. 
With the initial data shown in Table \ref{eq:2BInitialDatum} and $\Omega_1 = [0,1], ~ \Omega_2 = [0, 0.3]$, this rapid fall in the pool concentration in neurite 2 happens between $t=26$ and $t=36$, see Figure \ref{fig:2BPoolConcentration} (b).
In the case where the initial length is nearly similar this feature does not occur, see Figure \ref{fig:2BPoolConcentration} (a).

\section{Conclusion}

\subsection{Discussion}
Different experimental approaches have shown that a neurite has to extend beyond a minimal length to become an axon (\cite{dotti_experimentally_1987}, \cite{goslin_experimental_1989}, \cite{yamamoto_differential_2012}). 
During axon specification, intracellular transport is polarized towards the nascent axon to allow its extension (\cite{schelski_neuronal_2017}, \cite{bradke_neuronal_1997}), where several molecular mechanisms have been proposed for the length-dependent specification of axons (\cite{albus_cell_2013}, \cite{rishal_motor-driven_2012}, \cite{schelski_neuronal_2017}). 
Our live cell imaging experiments indicate that the flow of \gls{vesicle}s into a neurite increases when it extends indicating that \gls{vesicle} transport rates depend on changes in neurite length. 
The transport along \glspl{microtubul} connects pools of \gls{vesicle}s in the cell body and at the tip of neurites. 
The number of \gls{vesicle}s in the pool at the tip of the neurites reflects their growth potential because it provides material for membrane expansion (\cite{pfenninger_regulation_2003}, \cite{pfenninger_plasma_2009}). 
Our simulations show that the number of \gls{vesicle}s in the pool is higher in the longer neurite of the model neuron once the length advantage of the longer neurite has exceeded a certain length.
In addition, the simulations show oscillations in the concentration in the growth cones that were also observed in polarizing neurons \cite{JACOBSON2006797,10.7554/eLife.12387,doi:10.1038/msb.2010.51}. The length-dependent effect on the size of \gls{vesicle} pool indicates that the coupling by bidirectional transport adjusts transport rates to neurite length. 

Although the model presented here does not capture all aspects of \gls{vesicle} transport and neurite extension, it suggests that the bidirectional transport of \gls{vesicle}s between \gls{soma} and \gls{growth cone} couples the different \gls{vesicle} pools in a way that results in a preferential transport into the growing neurite.

\subsection{Outlook}
Motivated by the fact that the length of neurites is changing during the \gls{polarization} process, the most urgent extension of our model is the feature of a growing and shrinking domain.
Therefore we have to consider a free boundary value problem, i.e. the lengths of the domains $\Omega_1(t)$ and $\Omega_2(t)$ become time dependent.

The difficulty arising in this method is the following: 
In reality, growing and shrinking is a continuous process but numeric simulations are always discrete. 
Thus in numeric simulations the domain has to shrink by segments but currently the approximation of what happens to the \gls{vesicle}s located on these intervals is not obvious as in reality situations like these never occur.

Furthermore, as neurites grow by \gls{exocytosis} of \gls{vesicle}s in the membrane, a production term of \gls{vesicle}s in the \gls{soma} is necessary since at present the total mass of \gls{vesicle}s is constant which prevents the neurite from intensive growth which requires a huge amount of vesicles. 
Furthermore, as there is a maximal concentration in the neurites it can happen that two waves of particles are pushing onto each other resulting in traffic jams that are biological not meaningful.
Consequently, we have to choose the values of the production term carefully in a way that jams around the \gls{soma} are prevented.
 
The second feature, that would be suggestive to include, is an age-based population structure, i.e. the probability of \gls{vesicle}s leaving the pool increases with the length of the duration of its stay in the pool. 
Currently \gls{vesicle}s that enter a pool can immediately leave it in the next time, but this additional delay could result in concentration oscillations in the \gls{growth cone}s that reflect the cycles of stochastically occurring periods of extension and retraction of neurites mentioned in the introduction.

Finally, as pointed out in remark \ref{rem:Analysis}, challenging analytical problems arise in the context of this model.

\vspace{2ex}
\noindent \textbf{Acknowledgments:} The authors acknowledge support by EXC 1003 Cells in Motion and EXC 2044 Mathematics Münster, Clusters of Excellence, Münster, funded by the German science foundation DFG. pEGFP VAMP2 was a gift from Thierry Galli (Addgene plasmid \# 42308; \url{http://n2t.net/addgene:42308}; RRID:Addgene\_42308). We thank Ulrich M\"uller (Scripps Research Institute, La Jolla, CA, USA) for pDcx-iGFP.
Furthermore the authors would like to thank Martin Burger (FAU) for several interesting discussions.

\appendix 
\section{Appendix}

\subsection{Pseudocode and Computing Time}
\label{Pseudocodes and Computing Time}

For a better understanding of the numerics, we give a small overview on the computing time and how the implementation works. 

As we only analyzed a one dimensional numerical problem, the solving algorithm is for sure very fast. 
The elapsed time for the algorithm for model neuron with pools was about 6 minutes for $T = 100 $ on MATLAB R2017b.
For a better understanding of the code a pseudocode is given in Algorithm \ref{pseudocode}. 
\begin{algorithm}
\DontPrintSemicolon
\SetAlgoLined
\SetKwInOut{Input}{Input}
\SetKwInOut{Init}{Init}
\vspace{1ex}
\Input{Typical values for all parameters; initial pool concentrations and their maximum capacity; influx values and outflux velocities; potentials $V_a, V_r$; parameter $\epsilon$; initial concentration of anteo- and retrograde moving vesicles in both neurites}
\Init{Grid on space and time discretisation; Initialize each neurite as a structure array $N_i$ consisting of its initial vesicle concentrations; influx rates and outflux velocities; initial values of neighbouring pools, empty array for pool development;}
\vspace{1ex}
Calculate Scaling Parameters; \\
Plot initial concentration
\\
\vspace{1ex}
\For{every time step}{
	 Update concentrations in $N_i$ with the particle-hopping algorithm;
	 \\
	\For{every $1000^{th}$ time step}{
		Update each figures that shows the current vesicles density in a neurite or a pool;}
	Save current pool concentrations in an array;	
	\\
	Update concentrations in the pools $\Lambda_{\text{som}}$ and $\Lambda_{N_i}$; 
} 
\vspace{1ex}
Plot development in the pools;
\caption{Solving the time evolution of the vesicle concentrations in the model neuron}
\label{pseudocode}
\end{algorithm}

\printnoidxglossaries % iterate over all indexed entries
%\printnoidxglossary[sort=word]% main glossary
%\printnoidxglossary[type=symbols,sort=use]% symbols glossary

\printglossary

\bibliographystyle{plain}
%\bibliography{bib/bibliography}

\end{document}